\documentclass[pra,twocolumn,floatfix,superscriptaddress,notitlepage,longbibliography]{revtex4-1}

\usepackage[normalem]{ulem}
\usepackage{graphicx, color, graphpap}  % Include figure files
\usepackage{enumitem}
\usepackage{amsmath, amssymb, amsthm}

 \usepackage{ragged2e}

\usepackage{multirow}
\usepackage[dvipsnames]{xcolor}
\definecolor{IndustrialBlue}{RGB}{29,88,167}
\usepackage[colorlinks=true,urlcolor=IndustrialBlue,citecolor=IndustrialBlue,linkcolor=IndustrialBlue]{hyperref}
\usepackage[T1]{fontenc}
\usepackage{algorithm}
\usepackage{algpseudocode}
\usepackage{array}[=2016-10-06]
\usepackage{comment}
\usepackage{verbatim}
\usepackage{mathtools}
\usepackage{titlesec}
\usepackage{float}
\usepackage{booktabs}
\usepackage{cellspace}
\usepackage{soul}
\usepackage{multirow}
\usepackage{stmaryrd}
\usepackage{pgfplots}
\usepackage{placeins}
\usepackage[justification=raggedright, singlelinecheck=false]{caption}
\usepackage{subcaption}
\pgfplotsset{compat=1.18}

\usepackage{tikz}
\usetikzlibrary{shapes.geometric, arrows.meta}
\usetikzlibrary{fit,positioning}
\tikzstyle{item} = [rectangle, rounded corners, minimum width=3cm, minimum height=1cm,text centered, draw=black]
\tikzstyle{arrow} = [thick,->,>=stealth]
\usetikzlibrary{shapes.multipart}
\tikzset{
  item/.style={draw, rectangle, rounded corners, align=center, minimum width=2.5cm, minimum height=1cm},
  highlight/.style={draw=cyan, thick, dotted, rounded corners, inner sep=0.5cm}
}

\long\def\ca#1\cb{} %Use for commenting out: \ca...\cb

\newcommand{\ket}[1]{|#1\rangle}               %ket
\usepackage{circuitikz}
\usetikzlibrary{arrows, shapes.gates.logic.US, calc}
\usetikzlibrary{arrows,decorations.pathreplacing}
\usetikzlibrary{backgrounds,fit,decorations.pathreplacing,calc}

\newcommand{\GHZ}{\text{GHZ}}

\newcommand{\HC}{\mathcal{H}}

\newcommand{\Tr}{{\rm Tr}}

\renewcommand{\geq}{\geqslant}
\renewcommand{\leq}{\leqslant}
\renewcommand{\ge}{\geqslant}
\renewcommand{\le}{\leqslant}

\renewcommand{\Re}{\text{Re}}

\newcommand{\psitilde}{\widetilde{\smash{\psi}}}

\newcommand{\ltg}{{L^\tau_{\mathrm{global}}}}
\newcommand{\ltl}{{L^\tau_{\mathrm{local}}}}

\newcommand{\R}{\mathbb{R}}
\newcommand{\C}{\mathbb{C}}

\let\Re\relax

\DeclareMathOperator{\Re}{Re}
\newcommand{\wt}{\widetilde}

\DeclareMathOperator{\tr}{tr}
\DeclareMathOperator*{\E}{\mathbb{E}}
\newcommand{\Eunif}{{\mathcal{E}_{\mathrm{unif}}}}

\newcommand{\Bgammaba}{{{\mathbf{\Gamma}}_{BA}}}
\newcommand{\D}{{\mathbf{D}}}

\newtheorem{theorem}{Theorem}
\newtheorem{corollary}[theorem]{Corollary}
\newtheorem{lemma}[theorem]{Lemma}

\newtheorem{proposition}[theorem]{Proposition}

\newtheorem{definition}{Definition}

\renewcommand{\bm}[1]{\boldsymbol{#1}}            % AS: TODO temporary command
\newcommand{\var}{\textnormal{Var}}

\makeatletter
\renewcommand\thesection{\Roman{section}}
\renewcommand\p@subsection{\thesection.}

\renewcommand\p@subsubsection{\thesection.\Alph{subsection}.}

\makeatother

\allowdisplaybreaks
\begin{document}
\title{Localizing entanglement in high-dimensional states}
\author{Christopher Vairogs}
\affiliation{Department of Physics, University of Illinois Urbana-Champaign}
\author{Akanksha Chablani}
\affiliation{Department of Mathematics, University of Illinois Urbana-Champaign}
\author{Leo Lee}
\affiliation{Department of Mathematics, University of Illinois Urbana-Champaign}
\author{Hanyang Sha}
\affiliation{Department of Mathematics, University of Illinois Urbana-Champaign}
\author{Abigail Vaughan-Lee}
\affiliation{Department of Physics, University of Illinois Urbana-Champaign}
\affiliation{Department of Mathematics, University of Illinois Urbana-Champaign}

\author{Jacob L. Beckey}
\affiliation{Department of Mathematics, University of Illinois Urbana-Champaign}
\affiliation{Illinois Quantum Information Science and Technology Center (IQUIST), University of Illinois Urbana-Champaign}

\begin{abstract} 
In this work, we study the asymptotic behavior of protocols that localize entanglement in large multi-qubit states onto a subset of qubits by measuring the remaining qubits. We use the maximal average $n$-tangle that can be generated on a fixed subsystem by measuring its complement -- either with local or global measurements -- as our key figure of merit. These quantities are known respectively as the localizable entanglement (LE) and the entanglement of assistance (EA). We build upon the work of Ref.~\cite{Vairogs2024} that proposed a polynomial-time test, based on the EA, for whether it is possible to transform certain graph states into others using local measurements. We show, using properties of the EA, that this test is effective and useful in large systems for a wide range of sizes of the measured subsystem. In particular, we use this test to demonstrate the surprising result that general local unitaries and global measurements will typically not provide an advantage over the more experimentally feasible local Clifford unitaries and local Pauli measurements in transforming large linear cluster states into GHZ states. Finally, we derive concentration inequalities for the LE and EA over Haar-random states which indicate that the localized entanglement structure has a striking dependence on the locality of the measurement. In deriving these concentration inequalities, we develop several technical tools that may be of independent interest.
\end{abstract}

\maketitle

\section{Introduction}

Many quantum protocols rely on having access to highly-entangled multipartite states of a particular form. A common approach to creating such states is to measure and discard part of a larger quantum system in a more easily-prepared state. Broadly speaking, we say that such protocols attempt to \textit{localize} entanglement on the target system. Protocols in this vein are relevant across many fields of quantum information science, including measurement-based quantum computation~\cite{raussendorf2001oneway, raussendorf2003measurement, Gross2007BeyondOneWay}, state preparation via dynamic quantum circuits~\cite{Lu2022, Piroli2021}, and random state generation~\cite{Cotler2023, Ippoliti2023DynamicalPurification, Varikuti2024unravelingemergence}. The aim of this paper is to study the ability of protocols to localize entanglement onto a subsystem in the limit of \textit{large total system size}. 

To do so, we will rely on recently proposed variants of the entanglement of assistance (EA) and localizable entanglement (LE)~\cite{Verstraete-LE, Popp-LE, sadhukhan2017multipartite, banerjee2021localizing, harikrishnan2023localizing, amaro2018estimating, amaro2020scalable,banerjee2020uniform,banerjee2022hierarchies,krishnan2023controlling} as key figures of merit. The entanglement of assistance is defined as the maximum average value of entanglement, as measured by a prespecified entanglement measure, between a collection of subsystems that may be produced by performing measurements on all other subsystems. The localizable entanglement is defined analogously, but with the requirement that the measurements over subsystems must be local. Recent work~\cite{Vairogs2024} has proposed to use the LE and EA defined by the multipartite entanglement measure known as the $n$-tangle~\cite{Wong2001, Jaeger2003Invariance} as benchmarks for entanglement localizing protocols. Following up on this prior research, we take this approach in our study. Thus, we will refer to the LE and EA defined in terms of the $n$-tangle simply as the LE and EA, respectively. 

Herein, we are primarily concerned with the behavior of the LE and EA over large multipartite systems due to the importance of scalability in quantum science. Furthermore, due to measure concentration phenomena, in which the values of nicely behaved functions over large data sets deviate from a constant value with only small probability, we may avoid the intricacies of analyzing the LE and EA on particular states. Thus, we opt instead for an understanding of how these benchmarks behave over \textit{typical} states. In particular, we will consider graph states and Haar-random states over many qubits.

Graph states have broad applications across quantum information processing tasks, and so are a natural case to consider. In this paper, we argue how we may interpret the values that the EA assumes on graph states in terms of a simple and efficient test~\cite{Vairogs2024} for whether it is possible to transform a single copy of a given graph state into a graph state of maximal $n$-tangle using a restricted set of operations. The allowed operations we will consider consist of one round of arbitrary measurements on a prespecified collection of qubits followed by arbitrary local unitaries. From this point of view, we show that computing the average value of the EA across an ensemble of random graph states is equivalent to finding the probability that this test is conclusive and, hence, useful. Consequently, we may view the typical values of the EA simultaneously as a reflection of our ability to localize entanglement in large graph states but also as a benchmark of a criterion for deciding the possibility of certain graph-state manipulations. We then rigorously derive an approximation for the mean EA across an ensemble of random graph states. Using this approximation, we demonstrate that the graph transformation test is useful in deciding whether it is possible to extract a fixed state of maximal $n$-tangle, such as a GHZ state, when fewer than half the qubits are measured out and provide an interpretation for the case that more than half are measured. 

As a special case, we restrict our analysis to linear cluster states, an important subclass of graph states for applications. Building off of prior work that characterizes the ways in which it is possible to extract a GHZ state from a linear cluster state with local Clifford unitaries and local Pauli measurements~\cite{deJong2024}, we use the aforementioned graph state transformation test to study when general local unitaries and global measurements provide an advantage in extracting GHZ states. For this problem, we are particularly concerned with the idea of a  \textit{measurement configuration}, which is simply a selection of a qubit subsystems to be measured. We show that in almost all measurement configurations on a linear cluster state in which it is impossible to extract a GHZ state using local Pauli measurements and local Clifford unitaries, it is simultaneously impossible to do so using global measurements and general local unitaries in the limit of large system size.

Finally, to study the generic behavior of the LE and EA, we derive concentration inequalities for Haar-random states. We show that when less than half the qubits of a many-qubit system are measured, the values of the LE and EA cluster near their minimal value of zero. On the contrary, when more than half the qubits of a large system are measured, the typical values of the LE lie near zero, while the EA clusters near its maximal value of one. Thus, these concentration inequalities reveal a separation between the amount of entanglement that local and global measurements may localize in the limit of large system size and also a sharp dependence on the relative size of the measured subsystem. Furthermore, we discuss how these concentration inequalities bound our ability to extract states of maximal $n$-tangle, such as GHZ states, from typical states. 

This paper is organized as follows. We begin with some essential facts about the LE and EA defined in terms of the $n$-tangle in Section~\ref{sec:localizing-entanglement} before discussing graph states and the test for graph state transformations described above in Section~\ref{sec:graph-state-background}. Section~\ref{sec:graph-overview} establishes a connection between the mean value of the EA and the probability the test is conclusive. We compute this probability and provide supporting numerics in Sections~\ref{sec:asymp-sol-prob} and~\ref{sec:numerical-pr-sol}, respectively. Section~\ref{sec:concentration} provides rigorous concentration inequalities for the LE and EA over Haar-random states. Finally, further implications and open problems are discussed in Section~\ref{sec:future-directions}.

\section{Background}\label{sec:background}

\begin{figure}[hbtp] \centering
    \begin{tabular}{|l|l|}
        \hline 
        Variable & Definition \\
        \hline
        $\mathcal{H}$ & Hilbert space \\
        $| \Psi \rangle$ & multi-qubit state \\
        $d_A$ & dimension of $\mathcal{H}_A$\\
        $N_A$ & number of qubits in $\mathcal{H}_A$\\
        $\tau (| \Psi \rangle)$ & the $n$-tangle of $| \Psi \rangle$ \\
        $L^{\tau}_{\rm local}(| \Psi \rangle)$ & LE of $| \Psi \rangle$  w.r.t. the $n$-tangle\\
        $L^{\tau}_{\rm global} ( | \Psi \rangle )$ & EA of $| \Psi \rangle$ w.r.t. the $n$-tangle \\
        $\mathcal{C}(\mathcal{H})$ & collection of all ordered $\perp_N$ bases of $\mathcal{H}$ \\ 
        $\mathcal S(\mathcal H)$ & set of all normalized states in $\mathcal H$ \\
        $\mu_H$ & the Haar measure \\
        \hline
    \end{tabular}
    \caption{\centering Table of notation used throughout the text.}
    \label{table:example_table_figure} 
\end{figure}

\subsection{Localizing entanglement}\label{sec:localizing-entanglement}

Since we are interested in understanding protocols that localize entanglement in large systems, we want a way of quantifying entanglement across multiple subsystems, rather than across a particular bipartition, as one might with the entanglement entropy or generalized bipartite concurrence~\cite{Rungta2001, Mintert2005}. From an entanglement-theoretic viewpoint, a function of a quantum state that suitably quantifies some essential features of entanglement must \textit{(i)} vanish on product states and \textit{(ii)} remain non-increasing on average under local operations and classical communication (LOCC). Such functions are known as \textit{entanglement measures}. 

In this paper, we will use the entanglement measure known as the $n$-\textit{tangle} to quantify entanglement in multi-qubit states. For an $n$-qubit pure state $|\psi\rangle$, we will define the \textit{spin-flipped state} $|\psitilde\rangle$ as
\begin{equation}
    |\psitilde\rangle \coloneqq \sigma_y^{\otimes n}|\psi^\star\rangle,
\end{equation}
where $\sigma_y$ is the Pauli $y$-matrix and $|\psi^\star\rangle$ denotes the complex conjugate of $|\psi\rangle$ with respect to the computational basis. We will define the $n$-tangle $\tau(|\psi\rangle)$ of an $n$-qubit pure state $|\psi\rangle$ as
\begin{equation}~\label{eq:n-tangle-def}
    \tau(|\psi\rangle) \coloneqq |\langle \psi|\psitilde\rangle|.
\end{equation}
While the $n$-tangle is typically defined in the literature as the square of the quantity appearing in~\eqref{eq:n-tangle-def}, the above function is still a legitimate entanglement measure, as discussed in~\cite{Vairogs2024}. We choose to use the definition in~\eqref{eq:n-tangle-def} because it generalizes the well-known \textit{concurrence}~\cite{Wootters1997, Wooters1998EOF} in the sense that the above-defined $n$-tangle is equal to the concurrence when $n=2$, and also because the results of~\cite{Vairogs2024}, which are crucial for our analysis, use this convention. It is also known that the $n$-tangle is generically zero on states over an odd number of qubits, so we will only evaluate the $n$-tangle on even-qubit systems.

Note that as an entanglement measure, the $n$-tangle has a multipartite nature since its definition does not stipulate a bipartite cut of the $n$-qubit system. While there are a plethora of known multipartite entanglement measures, many of which are tuned to detect certain properties in quantum states, the $n$-tangle has several unique advantages. First, it is convenient that it has a particularly compact definition and is simple to compute, requiring no solutions to optimization problems. As mentioned before, the $n$-tangle may also be seen as a multipartite generalization of the two-qubit concurrence. Furthermore, the $n$-tangle plays an important role in the framework of stochastic local operations and classical communication (SLOCC). If a state $|\psi\rangle \in (\mathbb{C}^2)^{\otimes n}$ can be interconverted with nonzero probability by LOCC with state $|\varphi\rangle \in (\mathbb{C}^2)^{\otimes n}$, then $|\psi\rangle$ is said to be equivalent under stochastic LOCC (SLOCC) to $|\varphi\rangle$. It is known that two states are SLOCC-equivalent if and only if one may be obtained from the other by applying invertible local linear operators~\cite{Dur2000}. Therefore, local $\mathrm{SL}(2, \mathbb{C})$ (determinant-one) operators describe SLOCC-equivalent states up to scalar multiples. Significantly, the $n$-tangle is invariant under such $\mathrm{SL}(2, \mathbb{C})$ operations~\cite{Jaeger2003Invariance} and, hence, may be useful in characterizing SLOCC-inequivalence. For instance, the $n$-tangle may distinguish between the SLOCC classes of the generalized GHZ and W states~\cite{Dur2000} in $n$-qubit systems due to the fact that $\tau(|\GHZ\rangle) = 1$ while $\tau(|\mathrm{W}\rangle) = 0$.

Having established the $n$-tangle as our quantifier of multipartite entanglement, we proceed with our goal of studying how entanglement may be localized in large systems.   
Consider an $N$-partite state $|\Psi\rangle \in \mathcal{H}_1 \otimes ... \otimes \mathcal{H}_N$, where $\mathcal{H}_1,...,\mathcal{H}_N$ are the Hilbert spaces of of the $N$ subsystems.  Suppose that a measurement is performed over some subsystem $A \subset [N] := \{1,...,N\}$. Let $B \subset [N]$ denote the complement of $A$. Let $\mathcal{H}_A$ and $\mathcal{H}_B$ denote the Hilbert spaces of $A$ and $B$, respectively. Define $N_A \coloneqq |A|$, $N_B \coloneqq |B|$, $N \coloneqq N_A + N_B$. Throughout this work, \textit{we demand that $N_B$ is even}, so that the $n$-tangle assumes non-vanishing values on states over the $B$ subsystem. Let also $\Pi_A$ be a rank-one projective measurement on subsystem $A$ whose outcomes lead to post-measurement states $|\psi_i\rangle_B$ on $B$ with probability $p_i$. The authors of~\cite{Vairogs2024} proposed to use the quantity 
\begin{equation}\label{eq:mea-def}
    \ltg(|\Psi\rangle_{AB}) \coloneqq \max_{\Pi_A} \sum_i p_i \tau(|\psi_i\rangle_B),
\end{equation}
where the maximization is over all rank-one projective measurements on subsystem $A$, as a benchmark for our ability to \textit{localize} entanglement. We will refer to $L^\tau(|\Psi\rangle_{AB})$ as the \textit{entanglement of assistance} (EA). Essentially, the EA is the maximal average amount of entanglement, as measured by the $n$-tangle, that can be localized on $B$ by performing projective measurements on $A$. Crucially, this definition allows for the maximal value in~\eqref{eq:mea-def} to be achieved by \textit{global measurements}, whose operators may be entangled across the $N_A$ subsystems of $A$. However, it can be forbiddingly difficult to implement such collective measurements experimentally. Consequently, it was also proposed in~\cite{Vairogs2024} to restrict the optimization of~\eqref{eq:mea-def} to \textit{local measurements}, whose operators factorize as a tensor product across the $N_A$ systems in $A$ to obtain an alternative practically-motivated benchmark:
\begin{equation}\label{eq:lme-def}
    \ltl(|\Psi\rangle_{AB}) \coloneqq \max_{\Pi_A} \sum_i p_i \tau(|\psi_i\rangle_B).
\end{equation}
We will refer to $\ltl(|\Psi\rangle_{AB})$ as the \textit{localizable entanglement} (LE) of $|\Psi\rangle_{AB}$. We note that these definitions of the localizable entanglement and entanglement of assistance differ from those in Refs.~\cite{Verstraete-LE,Divincenzo1998} in two important ways. First, our $B$ system may be composed of more than just two qubits -- an important extension for many promising applications -- and the entanglement measure over the $B$ subsystem is taken to be the $n$-tangle, as opposed to the concurrence~\cite{Wooters1998EOF,Rungta2001}. We note that there is not a unique notion of localizable entanglement and recent works have considered different entanglement measures~\cite{Vairogs2024} and the ability to localize other functionals of quantum states~\cite{du2025Certifying}.

A key observation of~\cite{Vairogs2024} allows us to compute the EA in a simple way using the quantum fidelity function. Recall that for density matrices $\rho, \sigma$ over the same Hilbert space, the (square-root) fidelity is defined as $F(\rho, \sigma) \coloneqq \Tr\left[\sqrt{\sqrt{\rho}\sigma \sqrt{\rho}}\right]$. Theorem 1 of~\cite{Vairogs2024} shows that for arbitrary multi-qubit state $|\Psi\rangle_{AB} \in \mathcal{H}_A \otimes \mathcal{H}_B$, we have
\begin{equation}\label{eq:ltg-is-F}
    \ltg(|\Psi\rangle_{AB}) = F(\Psi_B, \widetilde \Psi_B),
\end{equation}
where $\Psi_B \coloneqq \tr_A[\Psi]$ and for any density matrix $\rho$ over $\mathcal{H}_B$, we have $\tilde{\rho} \coloneqq \sigma_y^{\otimes N_B} \rho^\star \sigma_y^{\otimes N_B}$ with $\rho^\star$ denoting the complex conjugate of $\Psi_B$ with respect to the computational basis. 
We note that~\eqref{eq:ltg-is-F} requires nothing more than knowledge of the reduced state $\Psi_B$ to compute $\ltg(|\Psi\rangle_{AB})$, while the definition of $\ltg(|\Psi\rangle_{AB})$ involves a complicated optimization over all projective measurements. One of our main contributions is developing methods of proving concentration of $\ltg(|\Psi\rangle_{AB})$, despite not having a closed form for the expression. While localizable entanglement is well-defined for all quantum states, we will focus many of our results on graph states -- a class of states with many exciting potential applications. We now review the properties of graph states needed to understand our main results.

\subsection{Graph states}\label{sec:graph-state-background}

Graph states are a family of quantum states that have served as an important example for many quantum information-theoretic concepts while also playing a crucial role as a resource in measurement-based quantum computing~\cite{raussendorf2001oneway, raussendorf2003measurement}. A graph state corresponds to a simple graph where the vertices represent qubits and the edges represent entangling operations applied between the corresponding qubits. Thus, the graph that defines a graph state encodes its structure of entanglement. 

More formally, let $G = (V, E)$ be a graph with a vertex set $V$ representing the qubits, and an edge set $E \subseteq \{\{a, b\}: a, b \in V~\mathrm{and}~a\neq b\}$. Then the graph state $|G\rangle$ is defined as 
\begin{equation}
    |G\rangle := \displaystyle \prod_{\{a, b\} \in E} CZ_{ab} \ket{+}^{\otimes N}
\end{equation}
where $CZ_{ab}$ is the controlled-Z gate acting between qubits $a$ and $b$. We will sometimes refer to $G$ as the \textit{underlying graph} of $|G\rangle$. 

A simple rule for determining the EA values of graph states is provided in~\cite{Vairogs2024}. Consider the following construction for abitrary graph state $|G\rangle \in \mathcal{H}_A \otimes \mathcal{H}_B$ whose underlying graph has edge set $E$. Let $\Bgammaba$ denote the $N_B \times N_A$ block of the adjacency matrix of $G$ that describes the connectivity between the vertices corresponding to qubits in $A$ and $B$. So, for any qubits $a \in A$ and $b \in B$, we have $(\Bgammaba)_{ba} = 0$ if $\{a,b\} \notin E$ and $(\Bgammaba)_{ba} = 1$ if $\{a, b\} \in E$. Define $\D \in \mathbb{F}_2^B$ to be the binary vector whose $b$-th component is $1$ if $b \in B$ has an even degree in the subgraph $G - A$ induced by $B$ and $0$ otherwise. Lemma 10 of~\cite{Vairogs2024} shows that for arbitary graph state $|G\rangle \in \mathcal{H}_A \otimes \mathcal{H}_B$, we have $F(G_B, \tilde{G}_B) = 1$ if the matrix equation $\Bgammaba  \mathbf{x} = \D$ has a solution in $\mathbb{F}_2^A$ and $F(G_B, \tilde{G}_B) = 0$ otherwise. The relation~\eqref{eq:ltg-is-F} then gives us the following rule, which is obtained in Ref.~\cite{Vairogs2024}:
\begin{theorem}[Theorem 11, \cite{Vairogs2024}]\label{thm:ltg-in-terms-of-matrices}
    Assume $N_B$ is even. Let $|G\rangle \in \mathcal{H}_A \otimes \mathcal{H}_B$ be an arbitrary graph state. If the matrix equation $\Bgammaba \mathbf{x} = \D$ has a solution $\mathbf{x} \in \mathbb{F}_2^A$, then $\ltg(|G\rangle) = 1$. If no such solution exists, then $\ltg(|G\rangle) = 0$.
\end{theorem}

In other words, if $\Bgammaba \mathbf{x} = \mathbf{D}$ has no solution, then it is impossible to transform $|G\rangle$ into any state with nonzero $n$-tangle using any measurements on the $A$ subsystem. On the other hand, if a solution to $\Bgammaba \mathbf{x} = \D$ exists, then some (possibly global) projective measurement over $A$ will produce a state with an $n$-tangle value of one for every measurement outcome. 

On a related note, recent literature has addressed the problem of characterizing ways of transforming graph states into GHZ states. Building on previous work~\cite{Hahn2019, Mannalath2022}, the authors of~\cite{deJong2024} complete a characterization of all the strategies with which it is possible to extract a GHZ state from a linear cluster state using just Pauli measurements and local Clifford operations. In a similar vein,~\cite{Frantzeskakis2023} proposes a method to extract with nonzero error a GHZ state from a linear cluster state prepared with a particular coherent error using local unitaries and local projective measurements. Since the satisfiability of the matrix equation $\Bgammaba \mathbf{x} = \mathbf{D}$ yields insight into allowed graph transformations, analyzing the behavior of the LE and EA on graph states via this matrix equation is useful for understanding when it is possible to extract GHZ states and other states of maximal $n$-tangle from more general graph states.

\section{Main Results}
We are now ready to discuss our main results. The first set of results concerns concentration phenomena in graph states. 
\subsection{Concentration Phenomena in Graph States}

\subsubsection{Overview}\label{sec:graph-overview}

Suppose we wish to produce a \textit{target} graph state $|G'\rangle_B$ with a \textit{maximal} $n$-tangle value of one, such as the GHZ state, over the subsystem $B$ via measurements on the $A$ subsystem and local unitaries on $B$. In this scenario, we allow for the possibility of classical communication. This implies that the choice of local unitaries may depend on the outcome of the measurement.

An initial obstacle is the very large number of graph states over $AB$ from which we might consider extracting $|G'\rangle_B$. It is challenging to check numerically or otherwise whether any of the infinitely many possible measurements and local unitaries will yield $|G'\rangle_B$ from any candidate state when total system size is large. In general, it is an NP-complete problem to decide whether it is possible to deterministically transform one graph state into another using even just local Clifford gates,  local Pauli measurements, and classical communication~\cite{Dahlberg2020, Dahlberg2020-2}.

To remedy these problems, note that we cannot obtain $|G'\rangle_B$ from any \textit{source} graph state $|G\rangle_{AB}$ for which $L^\tau(|G\rangle_{AB}) = 0$ using even global measurements on $A$ and general local unitaries on $B$. Moreover, such a transformation is forbidden even probabilistically. Thus, Theorem~\ref{thm:ltg-in-terms-of-matrices} implies that we can eliminate any graph state $|G\rangle_{AB}$ with an unsatisfiable matrix equation $\Bgammaba \mathbf{x} = \D$ as a candidate for a $|G'\rangle_B$-extraction protocol. However, Theorem~\ref{thm:ltg-in-terms-of-matrices} does \emph{not} imply that a satisfiable matrix equation for $|G\rangle_{AB}$ means that we may extract $|G'\rangle_B$ from $|G\rangle_{AB}$. In this way, we have a test for the viability of extracting $|G'\rangle_B$ from $|G\rangle_{AB}$ using one round of (possibly global) measurements over $A$ followed by local unitaries on $B$, but with the caveat that the test has an inconclusive outcome. We will refer to this test as the \textit{matrix equation test}: 
\begin{figure}[htb] \centering
    \begin{tabular}{cc}
        $\Bgammaba\mathbf{x} = \D$\textit{ is unsatisfiable} &  \ \ $\Bgammaba \mathbf{x} = \D$ \textit{is satisfiable}\\ $\Downarrow$
        & \ \ $\Downarrow$ \\
        impossible to extract $|G'\rangle_B$ & \ \    inconclusive\\
        with nonzero probability & 
    \end{tabular}
\end{figure}

This test is convenient because checking the satisfiability of a linear binary matrix equation requires only~\emph{polynomial time}, which circumvents the NP-completeness of conclusively deciding whether a source graph state may be transformed into a target graph state. However, the matrix equation test will not be useful if it is inconclusive for most graph states over $AB$. Thus, in order for it to be effective at cutting down the candidate source states from which we might extract a target state, it must be unsatisfiable for a \textit{majority} of graph states over $AB$. 

Given this observation, a natural question arises:~\emph{What is the probability that a random graph state over $AB$ yields a satisfiable matrix equation $\Bgammaba \mathbf{x} = \D$ for large system sizes?} Note that this query may have a strong dependence on the particular distribution used to weight the graph states. To account for this, we will entertain several natural choices for this distribution. For a fixed selection of subsystems $A$, $B$ and for an ensemble $\mathcal{E} = \{(p_G, |G\rangle_{AB})\}$ of graph states over $AB$, we let $p_s(\mathcal{E})$ denote the probability that a graph state drawn from $\mathcal{E}$ yields a matrix equation with a solution. 

Before discussing particular ensembles, we emphasize that the question above offers a natural framework for studying our ability to localize entanglement in graph states. Indeed, by considering the dichotomous behavior of $\ltg(|G\rangle_{AB})$ implied by Theorem~\ref{thm:ltg-in-terms-of-matrices}, we see that for a given ensemble $\mathcal{E}$ of graph states over $AB$, 
\begin{equation}\label{eq:prob-ltg-equiv}
    p_s(\mathcal{E}) = \mathbb{E}_{|G\rangle \sim \mathcal{E}}[\ltg(|G\rangle_{AB})\rangle].
\end{equation}In other words, by determining the probability that $\Bgammaba \mathbf{x} = \D$ has a solution, we manage to say something about the typical value of $\ltg(|G\rangle_{AB})$ on graph states and about the utility of the matrix equation test at the same time.

In what follows, we compute $p_s(\mathcal{E})$ for various ensembles $\mathcal{E}$ of graph states both analytically and numerically. Of particular interest is the uniformly weighted ensemble of all graph states. Another important special case is an ensemble over all graph states designed to weight topologically equivalent graphs equally.

\subsubsection{Uniformly random ensemble}\label{sec:asymp-sol-prob}

The first significant result of our work on graph states is a rigorously-derived approximation for the probability of a solution over the uniformly weighted graph state distribution. Let us fix a particular bipartition $A|B$ of the $N$ qubits and let $\Eunif$ denote the ensemble of all graph states over $AB$ with uniform distribution. That is, each graph state from $\Eunif$ occurs with a probability of $2^{-N(N-1)/2}$ 
since the total number of $N$-vertex graphs is $2^{N(N-1)/2}$. Since each graph is weighted equally, $p_s(\Eunif)$ gives us a good overall picture of the effectiveness of the matrix equation test. 

\begin{theorem}\label{thm:solution-prob}
    Assume that $N_B$ is even. Let $r>0$ be arbitrary. We may approximate $p_s(\Eunif)$ as 
    \begin{align}\label{eq:pr-sol-bounds}
        \frac{d_A + 1}{(1+r)(d_A + d_B - 1)} - \varepsilon_1&\leq p_s(\Eunif) \\ \nonumber
        &\leq \frac{d_A + 1}{(1-r)(d_A + d_B - 1)} + \varepsilon_2,
    \end{align}
    where the error terms $\varepsilon_1$ and $\varepsilon_2$ are defined in terms of $d_A, d_B$, and $r$ by 
    \begin{align}
        \varepsilon_1 &\coloneqq \frac{d_Ad_B(d_A - 1)(d_B - 1)}{r^2(1+r)(d_A + d_B - 1)^3} \\
        \varepsilon_2 &\coloneqq \frac{(d_A - 1)(d_B - 1)}{r^2 (d_A + d_B - 1)^3}.  
    \end{align}
    In particular, for fixed $d_B$ and $r$, the errors scale as $\varepsilon_1 = O(1/d_A)$ and $\varepsilon_2 = O(1/d_A^2)$. Similarly, we have $\varepsilon_1 = O(1/d_B)$ and $\varepsilon_2 = O(1/d_B^2)$ for fixed $d_A$ and $r$.
\end{theorem}   
The complete proof may be found in Appendix~\ref{app:gs-concentration}. The basic idea of the proof is as follows.  By equation~\eqref{eq:prob-ltg-equiv}, the probability $p_s(\Eunif)$ is simply the expectation value of $L^\tau(|G\rangle_{AB})$ over $\Eunif$. By considering the behavior of the fidelity function on the marginals of graph states, we argue that equation~\eqref{eq:ltg-is-F} implies that $L^\tau(|G\rangle_{AB}) = \tr[G_B \tilde{G}_B]/\tr[G_B^2]$, where $G_B \coloneqq \Tr_A[|G\rangle\langle G|_{AB}]$. Thus, to compute the probability of a solution, we may compute the expectation value of this quotient. Using Chebyshev's inequality, we show that we can approximate the expectation value of the quotient $\tr[G_B \tilde{G}_B]/\tr[G_B^2]$ as the quotient of the expectation values of $\tr[G_B\tilde{G}_B]$ and $\tr[G_B^2]$ with an error term controlled by the variance of the purity $\tr[G_B^2]$. To compute these expectation values, we adapt a meticulous counting argument from~\cite{zhou2022} employed to evaluate expectation values of linear functionals over random graph states. Finally, we directly re-purpose a computation from~\cite{zhou2022} to write down the variance of the purity $\tr[G_B^2]$ and, hence, the error terms.

Theorem~\ref{thm:solution-prob} suggests that 
\begin{equation}\label{eq:pr-sol-approx}
    p_s(\Eunif) \approx \frac{d_A + 1}{d_A + d_B - 1},
\end{equation}
with the arbitrary constant $r$ controlling the accuracy of this approximation. To be more precise, the first terms of the LHS and RHS of~\eqref{eq:pr-sol-bounds} approximate $(d_A + 1)/(d_A+d_B - 1)$ for small $r$. However, the error terms $\varepsilon_1$ and $\varepsilon_2$ grow as $r$ tends toward zero, yet the increase in error may be counteracted by choosing the dimension of one subsystem to be much larger than the other.

We numerically observe in Fig.~\ref{fig:Random graphs approximation} that the approximation~\eqref{eq:pr-sol-approx} mimics the exact behavior of $\Pr(\mathrm{solution})$, with an especially good agreement for more modestly sized $N$. However, one may also check numerically that the theoretical error bounds from Theorem~\ref{thm:solution-prob} may be fairly loose for finite $N$. Thus, to cement the utility of the bounds, we use them to derive the asymptotic behavior of $p_s(\Eunif)$ when $N_A = \alpha N$ for some fixed fraction $\alpha \in (0, 1)$ as $N$ tends to infinity:
\begin{corollary}\label{cor:tensor-network}
    Assume $N_B$ is even. Fix $\alpha \in (0, 1)$. Then the following hold. 
    \begin{enumerate}
        \item If $\alpha < 1/2$, then $p_s(\Eunif) \to 0$ as $N \to \infty$ with $N_A = \lfloor \alpha N\rfloor$.
        \item If $\alpha > 1/2$, then $p_s(\Eunif) \to 1$ as $N \to \infty$ with $N_A = \lceil \alpha N \rceil$. 
    \end{enumerate} 
\end{corollary}
That is, the overwhelming majority of graph states over $AB$ will not yield a solvable matrix equation whenever $N_A < N/2$ for large total system size. Consequently, the matrix equation test is effective at narrowing down the possible candidates for the source state in a graph state extraction protocol targeting states with maximal $n$-tangle in this regime. While the matrix equation test \textit{does} \textit{not} allow us to conclude anything about whether a \textit{particular} graph state with maximal $n$-tangle may be extracted from a source state in the case that the matrix equation is solvable, Theorem~\ref{thm:ltg-in-terms-of-matrices} \textit{does} ensure that there exists some measurement over $A$ that yields \textit{some} maximal $n$-tangle state for all measurement outcomes. Thus point (2) from Corollary~\ref{cor:tensor-network} implies that we may extract a maximal $n$-tangle state with nonzero probability from the typical graph state $|G\rangle_{AB}$ in the regime $N_A > N/2$ for large $N$.

\subsubsection{Alternative graph ensembles} \label{sec:numerical-pr-sol}

While the approximation provided by Theorem~\ref{thm:solution-prob} is convenient for appraising the matrix equation test, a potential limitation is that the uniform distribution over all graph states unfairly biases certain kinds of graph states. The issue lies in the fact that certain families of isomorphic graphs contain more members than others, and hence are weighted more heavily by the uniform distribution. For instance, suppose $N = 3$ and consider the scenario reflected in Fig.~\ref{fig:equivalent-graphs}. Due to the symmetry of the two graphs in case 2, the distinction between their associated graph states is operationally irrelevant in a setting where we wish to extract states over $B$ via projective measurements on $A$. On the other hand, the state corresponding to the graph of case 1 presents a distinct experimental scenario for a state extraction. However, the structures reflected by case 1 and case 2 respectively occur with distinct probabilities of 1/8 and 2/8 in the uniform ensemble of graph states over $AB$. 

\begin{figure} \centering
        \includegraphics[width = 9cm]{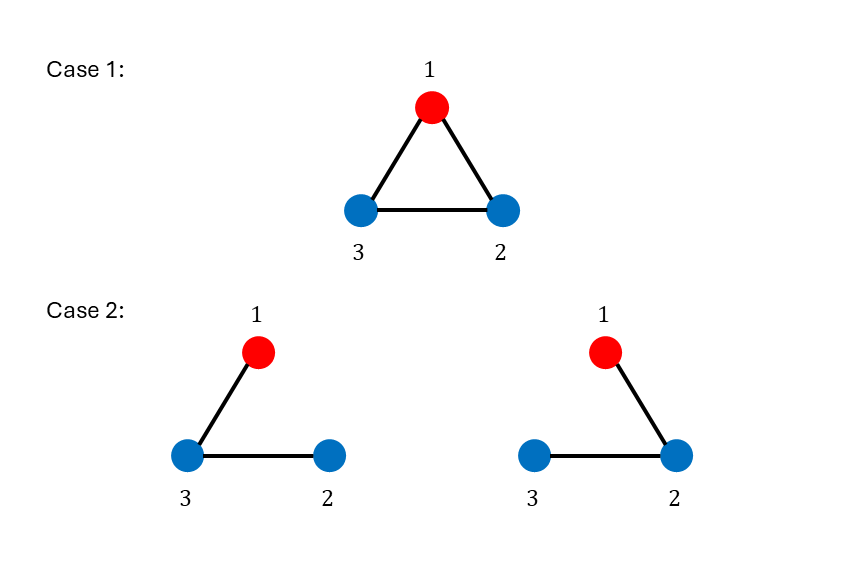}
        \caption{Different graph states can be equivalent for the purpose of entanglement localization when their underlying graphs are related by a graph isomorphism that preserves $A$ and $B$. In this figure, $A = \{1\}$ and $B = \{2, 3\}$.}
        \label{fig:equivalent-graphs}
\end{figure}

Another issue arises from the connectedness of the underlying graphs. 
Graph states with disconnected graphs may be written as product states over the subsystems corresponding to their connected components. A product state across a collection of subsystems will indeed remain a product states after local measurements and local unitaries. Hence, we generally wish to consider connected graphs.

Taking into account these issues, we consider an alternative distribution. We say that two \textit{connected} $N$-vertex graphs $G_1$ and $G_2$ whose vertices are respectively bipartitioned as $A_1|B_1$ and $A_2|B_2$ have isomorphic bipartitions if there exists some graph isomorphism between $G_1$ and $G_2$ that maps $A_1$ into $A_2$ and $B_1$ into $B_2$. Let us consider the ensemble $\mathcal{E}_{\mathrm{isom}}$ of all connected $N$-qubit graph states in which each graph graph state $|G\rangle_{AB}$ occurs with a probability of 
\begin{equation}
    p_G = \frac{1}{\left|[G]\right||I|},
\end{equation}
where $I$ denotes the set of all equivalence classes under isomorphisms of connected bipartitioned $N$-vertex graphs and $[G]$ refers to the equivalence class of $G$. Under this distribution, each equivalence class occurs with equal probability, so that no graph topology is biased over another. 

In   Fig.~\ref{fig:alternative-ensembles}a, we employ a numerical sampling method (see Appendix~\ref{app:noniso_sim}) to compute $p_s(\mathcal{E}_{\mathrm{isom}})$. It is significant that our approximation~\eqref{eq:pr-sol-approx} nearly coincides with $p_s(\mathcal{E}_{\mathrm{isom}})$ for the uniformly random isomorphism classes when $6\leq N \leq 10$. That the coincidence should be so tight is not clear \textit{a priori}. Furthermore, the close agreement suggests that our approximation~\eqref{eq:pr-sol-approx} is realistic even for a more practically motivated distribution of graph states.

On a separate note, we also define distributions in Appendix~\ref{app:graph-families} over particular families of graph states to understand the behavior of the matrix equation in special cases (see Fig.~\ref{fig:alternative-ensembles}b). In this approach, we find that the probability of a solution across many families behaves as a logistic curve with respect to $N_A$, echoing its behavior for the uniform distribution over isomorphism classes of bipartitioned graphs. Interestingly, when the total system size is fixed at $N = 16$, the probability of a solution over the distribution of bipartitioned 4-regular graphs is nearly identical to the analytic approximation for $p_s(\Eunif)$ given by~\eqref{eq:pr-sol-approx} for the uniform distribution over all graphs on $AB$.

\begin{figure}[t!]
    \begin{tikzpicture}
    \begin{axis}[
        xlabel={$N_A$},
        ylabel={$p_s(\Eunif)$},
        legend style={at={(1.0, 0)},anchor=south east, scale = 0.6},
        grid=major,
        width=.47\textwidth,
        height=6cm,
        ymin=0, ymax=1.05,
        xmin=0, xmax=19, 
        % xtick distance=0.2
    ]
    
    % ---- Dataset 1 ----
    \addplot[blue, thick, only marks, mark=*] coordinates {
    (10*0.2,0.016) (10*0.4,0.173) (10*0.6,0.896) (10*0.8,0.993)
    };
    \addlegendentry{$N = 10$}

    \addplot[red, thick, only marks, mark=*] coordinates {
    (15*0.06666667,0.001) (15*0.2,0.002) (15*0.333333,0.01) (15*0.4666667,0.249) (15*0.6,0.975) (15*0.733333,1.0) (15*0.866667,1.0)
    };
    \addlegendentry{$N=15$}

    \addplot[green!60!black, thick, only marks, mark=*] coordinates {
    (20*0.1,0.0) (20*0.2,0.0) (20*0.3,0.0) (20*0.4,0.004) (20*0.5,0.783) (20*0.6,0.997) (20*0.7,0.999) (20*0.8,1.0) (20*0.9,1.0)
    };
    \addlegendentry{$N = 20$}
    
    \addplot[gray, thick, only marks, mark=*] coordinates {
    (25*0.04,0.0) (25*0.12,0.0) (25*0.2,0.0) (25*0.28,0.0) (25*0.36,0.0) (25*0.44,0.005) (25*0.52,0.976) (25*0.6,0.997) (25*0.68,1.0) (25*0.76,1.0) (25*0.84,1.0) (25*0.92,1.0)
    };
    \addlegendentry{$N=25$}
    
    \addplot [
        domain=0:21, 
        samples=100, 
        color=blue,
        ultra thick,
        dashed
        ]
        {(2^x+1)/(2^x + 2^(10 - x) - 1)};
    \addplot [
        domain=0:21, 
        samples=100, 
        color=red,
        ultra thick,
        dashed
        ]
        {(2^x+1)/(2^x + 2^(15 - x) - 1)};

    \addplot [
        domain=0:21, 
        samples=100, 
        color=green!60!black,
        ultra thick,
        dashed
        ]
        {(2^x+1)/(2^x + 2^(20 - x) - 1)};
 
    \addplot [
        domain=0:21, 
        samples=100, 
        color=gray,
        ultra thick,
        dashed
        ]
        {(2^x+1)/(2^x + 2^(25 - x) - 1)};    
    
    \end{axis}
    \end{tikzpicture}
    \caption{Probability $p_s(\Eunif)$ of finding a solution for the uniformly weighted ensemble against $N_A$ in graphs with $N = 10, 15, 20$ and $25$ vertices. Data points along solid lines represent numerical estimates of $p_s(\Eunif)$ obtained from samples of 1000 random graphs drawn from $\Eunif$. The dashed lines indicate our approximation from equation~\ref{eq:pr-sol-approx}.}
    \label{fig:Random graphs approximation}
\end{figure}
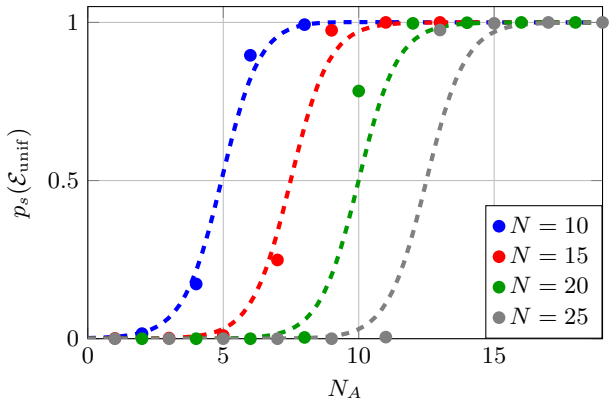

\subsubsection{Linear Cluster States}

We now turn towards analyzing the solvability of the matrix equation on linear cluster states, which have a line graph as their underlying graph. The authors of~\cite{deJong2024} previously characterized all the possible ways of deterministically extracting a GHZ state from a linear cluster state using local Pauli measurements, local Clifford operations, and classical communication (LC + LPM + CC). The restriction to LC + LPM + CC is natural because these operations are more experimentally feasible, and also because it makes the problem more mathematically tractable. However, it is unclear whether an expanded set of allowable operations will allow for more possible GHZ state extraction schemes. This leads us to ask whether such an advantage is offered by general local unitaries, global measurements over the measured system $A$, and classical communication (LU + GM + CC). Once again, the presence of classical communication implies that local unitaries may depend on measurement outcomes. 

Our next result answers this inquiry in the negative for large $N$. To formulate this result, we establish the following notation. We will refer to a particular assignment of the subsystems $A$ and $B = [N]\setminus A$ to the total $N$-qubit system as a \textit{measurement configuration} since $A$ labels the measured subsystem. For instance, $A = \{1, 3\}, B = \{2, 4, 5\}$ and $A = \{3, 5\}, B = \{1,2, 4\}$ would be considered two distinct measurement configurations for $N = 5$. The $N$-qubit linear cluster state $|L\rangle_{AB}$ is the graph state
\begin{equation}
    |L\rangle_{AB} \coloneqq \bigotimes_{i=1}^{N-1} \mathrm{CZ}_{i,i+1} |+\rangle^{\otimes N} 
\end{equation}
and the $N$-qubit GHZ state over the $B$ subsystem is
\begin{equation}
    |\GHZ\rangle_B \coloneqq \frac{1}{\sqrt{2}}\left(\bigotimes_{i \in B} |0\rangle_i + \bigotimes_{i\in B} |1\rangle_i\right).
\end{equation}

\begin{theorem} \label{thm:dejong-accordance}
    Let $S_N$ denote the set of measurement configurations on $N$ qubits for which $N_B$ is even and it is impossible to deterministically transform the linear cluster state $|L\rangle_{AB}$ into $|\GHZ\rangle_B$ via LC + LPM + CC. Let $T_N$ denote the set of measurement configurations on $N$ qubits for which $N_B$ is even and it is impossible to transform $|L\rangle_{AB}$ into $|\GHZ\rangle_B$ via LU + GM + CC with \textit{nonzero probability}, so that $T_N \subseteq S_N$. Then
    \begin{equation}
        \lim_{N \to \infty} \frac{|T_N|}{|S_N|} = 1. 
    \end{equation}
\end{theorem}
A proof of this theorem may be found in Appendix~\ref{app:extracting-ghz}. The proof relies on a comparison between the graph theoretic criteria for GHZ extraction outlined in~\cite{deJong2024} and the matrix equation criterion of Theorem~\ref{thm:ltg-in-terms-of-matrices} together with a combinatroial counting argument. 

Theorem~\ref{thm:dejong-accordance} implies that in almost every case (with even $N_B$) where it is impossible to deterministically extract a GHZ state using LC + LPM + CC, it will be impossible to do so with \textit{even non-zero probability} using the significantly more expanded operations of LU + GM + CC. Therefore, LU + GM + CC does not provide an advantage over LC + LPM + CC in the task of extracting GHZ states from linear cluster states asymptotically. Furthermore, since our matrix equation test detects every case in which it is impossible to stochastically extract a GHZ state via LU + GM + CC, Theorem~\ref{thm:dejong-accordance} implies that our matrix equation test strengthens in the asymptotic regime the conclusions of the test from~\cite{deJong2024} for GHZ extraction from linear cluster states under LC + LPM + CC.

\begin{figure*}
\centering
\begin{tabular}{cc}
    \includegraphics[width=0.47\linewidth]{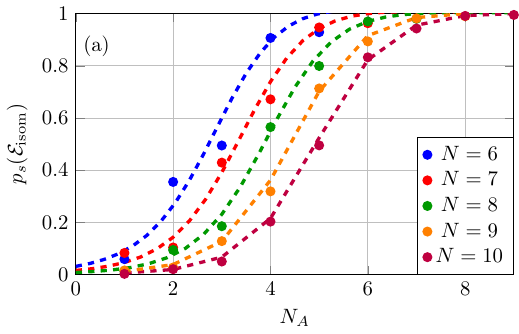}&
    \includegraphics[width=0.47\linewidth]{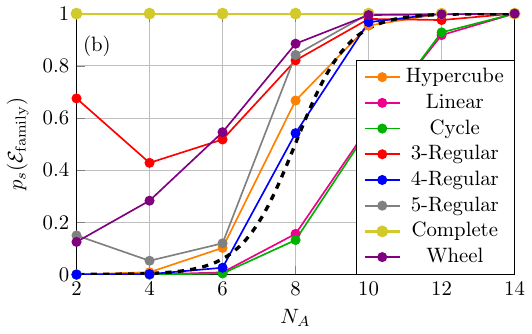}\\[2\tabcolsep]
    
\end{tabular}
\caption{(a) Probability $p_s(\Eunif)$ of finding a solution for the uniformly weighted ensemble against $N_A$ in graphs with $N = 10, 15, 20$ and $25$ vertices. Data points along solid lines represent numerical estimates of $p_s(\Eunif)$ obtained from samples of 1000 random graphs drawn from $\Eunif$. The dashed lines indicate our approximation from equation~\ref{eq:pr-sol-approx}. (b) Probability of finding a solution against $N_A$ for graphs drawn from a uniform distributions over graph families with various structural properties. Here, $N = 16$. The dashed line is the approximation we calculate in \eqref{eq:pr-sol-approx}. Note that the approximation is almost identical to the probability of getting a solution for a $4$-regular graph.}
\label{fig:alternative-ensembles}
\end{figure*}

\subsection{Concentration Phenomena in Haar Random States}\label{sec:concentration}
We now turn to a study of the average values of the LE and EA when the input state is sampled uniformly with respect to the uniform probability measure on $\HC_A \otimes \HC_B$. One of the principal aims of this work is to understand our ability to localize entanglement in high-dimensional systems, with $\ltl$ and $\ltg$ as our metrics of choice. Ref.~\cite{Vairogs2024} established the fact that $\ltg$ concentrates to near-maximal values whenever $d_A \gg d_B$. However, the behavior of $\ltg$ in the regime $d_B \gg d_A$ and the behavior of $\ltl$ for both $d_B \gg d_A$ and $d_A \gg d_B$ remained ill-understood. Understanding potential separations between the performance of local and global measurements for localizing entanglement is a matter of both theoretical and practical importance. However, as state above, because $\ltl$ does not have a closed form, standard concentration of measure techniques do not immediately apply. In this section, we develop methods of analytically determining the behavior of $\ltl$ as well as $\ltg$ in the regime $d_B \gg d_A$. 

\subsubsection{Primer on High-dimensional Probability}\label{sec:high-d-prob}

Concentration of measure is a powerful tool for analyzing the behavior of quantum state functionals in quantum information theory. Broadly speaking, it refers to the phenomenon that sufficiently well-behaved functions on high-dimensional Hilbert spaces are highly unlikely to deviate significantly from their mean value. When a state functional—such as the entanglement entropy—exhibits such concentration, its value for typical quantum states in large systems can be inferred from its average behavior~\cite{Hayden2006}. In this way, one can understand the typical properties of complex quantum systems without having to analyze each state individually.

Several concentration inequalities make these ideas precise by bounding the probability that a random variable deviates from its mean. Classic examples include Markov’s and Chebyshev’s inequalities which require control of the first and second moments of one's random variable, respectively. For functions on high-dimensional spheres—such as those describing random, pure quantum states—Lévy’s lemma provides a much stronger bound but requires a closed form for the mean and computation of the function’s Lipschitz constant. 

When standard concentration inequalities fail to give sufficiently tight bounds on the deviations of a functional, a standard approach is to utilize a so-called $\varepsilon$-net over quantum states. A collection $\mathcal{N}$ of (normalized) pure states from a given Hilbert space $\mathcal{H}$ is said to be an $\varepsilon$-\textit{net} if for any state $|\psi\rangle \in \mathcal{H}$, there exists a state $|\varphi\rangle \in \mathcal{N}$ such that $\| |\psi\rangle \langle \psi| - |\varphi\rangle \langle \varphi| \| \leq \varepsilon$. Here, $\| \cdot \|$ refers to a norm over $\mathcal{H}$, which is typically taken to be the 1-norm or 2-norm. It was shown in~\cite{Hayden_2004} that for any Hilbert space $\mathcal{H}$ with $\dim \mathcal{H} = d$, there exists an $\varepsilon$-net $\mathcal{N}$ of cardinality $|\mathcal{N}| \leq (5/\varepsilon)^{2d}$. 

To see how these arguments typically proceed, suppose that the probability that a state functional $f$ deviates from its mean value $\mu$ for random states $|\psi\rangle \in \mathcal{H}$ that simultaneously lie within the neighborhood of any fixed state $|\varphi\rangle$ is small, \textit{i.e.}, 
\begin{equation}\label{eq:joint-prob-concentration}
    \Pr_{|\psi\rangle}(|f(|\psi\rangle) - \mu| > \varepsilon~\mathrm{and}~\||\psi\rangle \langle \psi| - |\varphi\rangle\langle \varphi|\|< \varepsilon) \leq \delta.
\end{equation}
Then since any randomly sampled state $|\psi\rangle \in\mathcal{H}$ must lie in an $\varepsilon$-neighborhood of a state from an $\varepsilon$-net $\mathcal{N}$ with cardinality at most $(5/\varepsilon)^{2d}$, we may apply a union bound to get
\begin{equation}
    \Pr_{|\psi\rangle}(|f(|\psi\rangle) - \mu|> \varepsilon) \leq |\mathcal{N}|\delta. 
\end{equation}
Thus, if we can show that the $\delta$ bounding the probability in~\eqref{eq:joint-prob-concentration} decays sufficiently fast in relation to the scaling of the cardinality of $\mathcal{N}$, we can guarantee that $f(|\psi\rangle)$ deviates significantly from $\mu$ with exceedingly low probability. As it turns out, it is possible in certain notable cases to prove such a thing when it is challenging to prove a large deviation bound for $f$ directly~\cite{Hayden2006}.

With this intuition in place, we now mention what was known about the typical values of these localizable entanglement measures. The authors of~\cite{Vairogs2024} studied the typical behavior of the MEA. To state their result precisely, let us first let $\mu_H$ denote the measure on the pure states of $\mathcal{H}_A \otimes \mathcal{H}_B$ induced by the Haar measure~\cite{Mele2024} on the unitary group that acts on $\mathcal{H}_A \otimes \mathcal{H}_B$. When sampling random states with respect to this measure, the authors of~\cite{Vairogs2024} demonstrated that for arbitrary $\varepsilon >0$,
\begin{align}
    \Pr_{|\Psi\rangle \sim \mu_H}&\left(\ltg(|\Psi\rangle) \leq 1- \sqrt{2d_B/d_A} - \varepsilon\right) 
    \\
    &\leq 2\exp(-C d_Ad_B \varepsilon^2),
\end{align} 
where $C$ is an irrelevant pre-factor independent of dimension. Thus, when $d_A \gg d_B$, the values of $\ltg(|\Psi\rangle)$ are near maximal since the $n$-tangle assumes values between zero and one. In contrast, the typical behavior of the EA in the case that $d_A \ll d_B$ and the behavior of the LE in both regimes presented several analytical challenges that were not overcome until the present work. That said, our results were guided by preliminary numerical evidence presented in~\cite{Vairogs2024} which suggested that the LE values tend to a lower value than their maximum value of one when $d_A \gg d_B$.

\subsubsection{Results}

Recall that we have a system of $N$ qubits and two subsystems labeled by $A \subset [N]$ and $B = [N] \setminus A$. Furthermore, $\mathcal{H}_A$ and $\mathcal{H}_B$ label the Hilbert spaces of the $A$ and $B$ subsystems, respectively. We will write $d_A = 2^{N_A}$ and $d_B = 2^{N_B}$. For any Hilbert space $\mathcal{H}$, let $\mathcal{S}(\mathcal{H})$ denote the set of normalized states within $\mathcal{H}$. We also denote by $\mathcal{C}(\mathcal{H})$ the collection of all ordered orthonormal bases of Hilbert space $\mathcal{H}$. In this section, we will use $\mu_H$ to denote the measure induced on the set of pure states by the Haar measure of the unitary group acting on the relevant Hilbert space.

If a function $f$ takes a vector $|\psi\rangle$ as an input, we sometimes write $f(\psi)$ instead of $f(|\psi\rangle)$. For a given $|v\rangle \in \mathcal S(\mathcal H_A)$, suppose a Haar-random state $|\Psi\rangle \in \mathcal S(\mathcal H_A \otimes \mathcal H_B)$ undergoes a projective measurement on subsystem $A$ defined by projectors $\{|i\rangle \langle i|_A \otimes I_B\}$ for some $\{|i\rangle\} \in \mathcal C(\mathcal H_A)$ such that $|1\rangle = |v\rangle$. The probability, denoted by $p_v(\Psi)$, of obtaining the post-measurement state
\begin{align}
    \frac{(|v\rangle\langle v|_A \otimes I_B) |\Psi\rangle}{\sqrt{\langle \Psi | (|v\rangle\langle v|_A \otimes I_B) | \Psi \rangle}}
\end{align}
is given by 
\begin{equation}
    p_v(\Psi) := \langle \Psi | (|v\rangle\langle v|_A \otimes I_B) | \Psi \rangle.
\end{equation}
Since
\begin{equation}
    (|v\rangle\langle v|_A \otimes I_B)|\Psi\rangle = |v\rangle \otimes (\langle v|_A \otimes I_B)|\Psi\rangle,
\end{equation}
let us discard the qubits in subsystem $A$ and define
\begin{equation}
    |M_v(\Psi)\rangle := \frac{(\langle v|_A \otimes I_B)|\Psi\rangle}{\sqrt{p_v(\Psi)}}.
\end{equation}
One can check that $|M_v(\Psi)\rangle$ is indeed normalized and is in Hilbert space $\mathcal H_B$. Denote by $F_v(\Psi)$ the expected post-measurement $N_B$-tangle of $|\Psi\rangle$ associated with $|v\rangle$, or
\begin{equation} \label{eq:fvpvmv}
    F_v(\Psi) := p_v(\Psi) \tau(M_v(\Psi)) 
\end{equation}
Furthermore, given a basis $\beta := \{|\varphi_i\rangle\}_{i=1}^{d_A} \in \mathcal C(\mathcal H_A)$, if the set of projective measurement operators is described by $\{|\varphi_i \rangle \langle \varphi_i|_A \otimes I_B\}_{i=1}^{d_A}$, then we define the average post-measurement $N_B$-tangle of $|\Psi\rangle$ given basis $\beta$ to be
\begin{equation}
    \overline \tau_\beta(\Psi) := \sum_{i=1}^{d_A} \nolimits F_{\varphi_i}(\Psi).
\end{equation}
In this way, we may write the EA as
\begin{equation}
    \ltg(|\Psi\rangle) = \max_{\beta \in \mathcal C(\mathcal H_A)} \overline \tau_\beta (\Psi).
\end{equation}

Our first order of business is to upper bound the typical values of the EA $\ltg$ in the regime $d_B \gg d_A$. Let $|\varphi_{\max}\rangle \in \mathcal S(\mathcal H_A)$ be a state such that $F_v(\Psi) \le F_{\varphi_{\max}}(\Psi)$ for all $|v\rangle \in \mathcal S (\mathcal H_A)$ and define
\begin{equation}
    K := \sqrt{\tfrac{2}{d_B + 1}}. \label{eq:K}
\end{equation}
In Appendix~\ref{app:concentration-haar}, Lemma~\ref{lem:neighbor-bound-F-v} states that for an arbitrary $|v\rangle \in \mathcal S(\mathcal H_A)$, the value $F_{\varphi_{\max}}(\Psi)$ is \textit{typically} close to $K/d_A$ as long as $|v\rangle$ is close enough to $|\varphi_{\max}\rangle$. Moreover, by \citep[Lem.~II.4]{Hayden_2004}, there exists an $\varepsilon$-net $\mathcal N$ on $\mathcal H_A$ so that at least one state in $\mathcal N$ is close to $|\varphi_{\max}\rangle$. Combining the two results with a union bound styled argument as in Section~\ref{sec:high-d-prob}, we have the following theorem.

\begin{theorem} \label{thm:global-rB}
    For all $\varepsilon > 0$, the probability bound holds:
    \begin{multline} \label{eq:global-rB}
        \Pr_{|\Psi\rangle \sim \mu_H} \big( \ltg(\Psi) \ge K + \varepsilon \big) \\ \le 2\left(\tfrac{10\sqrt2 d_Ad_B}{\varepsilon}\right)^{2d_A} \exp\left(-\tfrac{d_B\varepsilon^2}{18\pi^3(4\sqrt2+2)^2d_A}\right)
    \end{multline}
\end{theorem}

With previous results from \cite{Vairogs2024}, we can draw the following conclusion.
\begin{corollary} \label{cor:global-concentration}
Let $\varepsilon, \delta>0$ be arbitrary. Then for any $d_A\geq 2$, there exists an $N_0 \in \mathbb{N}$ such that for all $d_B \geq 2^{N_0}$, we have 
\begin{equation}
    \Pr_{|\Psi\rangle \sim \mu_H} \big( \ltg(\Psi)\geq K + \varepsilon \big) \leq \delta.
\end{equation}

On the other hand, \citep[Thm.~8]{Vairogs2024} implies that for all $d_B \geq 2$, there exists an $N_0' \in \mathbb{N}$ such that for all $d_A \geq 2^{N_0'}$, 
\begin{equation}
    \Pr_{|\Psi\rangle \sim \mu_H} \big( \ltg(\Psi)\leq 1 - \sqrt{2d_B/d_A} - \varepsilon \big) \leq \delta.
\end{equation}
\end{corollary}

Next, we wish to compute an upper bound for the typical values of the LE in the regime $d_A \gg d_B$. To do so, we will use an $\varepsilon$-net construction similar to that described in Section~\ref{sec:high-d-prob}.

\begin{definition} \label{def:basis-norm}
    In Hilbert space $\mathcal{H}$ of dimension $d$, suppose that $\beta = \{|\varphi_i\rangle \}_{i=1}^d$ and $\gamma = \{|\eta_i\rangle\}_{i=1}^d$ are ordered orthonormal bases. Then the \textbf{basis norm} (or \textbf{$\bm B$-norm}) between bases $\beta$ and $\gamma$ is defined by
    \begin{equation}
        \|\beta - \gamma \|_B := \max_{1\leq i \leq d} \||\varphi_i\rangle\langle\varphi_i| - |\eta_i\rangle\langle \eta_i| \|_1.
    \end{equation}
\end{definition}

It is due to this definition that we take care to discuss \textit{ordered} bases throughout this note, as we want to only compare trace distances between basis elements of like index in our notion of basis distance. Next, given $\beta = \{|\varphi_i\rangle\}_{i=1}^{d_1}\in \mathcal{C}(\mathcal{H}_1)$ and $\gamma = \{|\eta_j\rangle\}_{j=1}^{d_2} \in \mathcal{C}(\mathcal{H}_2)$, we define $\beta \otimes \gamma \in \mathcal{C}(\mathcal{H}_1 \otimes \mathcal{H}_2)$ to be the ordered basis for $\mathcal{H}_1 \otimes \mathcal{H}_2$:
\begin{align}
    \beta \otimes \gamma :=\{&|\varphi_1\rangle |\eta_1\rangle, \dots, |\varphi_{d_1}\rangle |\eta_1\rangle, \notag \\
    &|\varphi_1\rangle |\eta_2\rangle, \dots, |\varphi_{d_1}\rangle |\eta_2\rangle, \notag \\ &\dots, \notag \\
    &|\varphi_1\rangle |\eta_{d_2}\rangle, \dots, |\varphi_{d_1}\rangle |\eta_{d_2}\rangle \}.
\end{align}
Finally, for any Hilbert spaces $\mathcal{H}_1, \dots, \mathcal{H}_n$ and $\mathcal H := \mathcal H_1 \otimes \dots \otimes \mathcal H_n$, define 
\begin{equation}\label{eq:p-h-def}
    \mathcal{P}(\mathcal{H}) \coloneqq \{\beta_1 \otimes \dots \otimes \beta_n: \beta_i \in \mathcal{C}(\mathcal{H}_i)\}.
\end{equation}
to be the collection of all orthonormal bases in $\mathcal H$ that ensure a tensor-product form on the qubits of systems corresponding to $\mathcal H_1$, $\dots$, $\mathcal H_n$. Obviously, $\mathcal{P}(\mathcal{H}) \subset \mathcal{C}(\mathcal{H})$.

\begin{definition} \label{def:basis-epsnet}
    Let $\varepsilon>0$ and let $\mathcal{H}  = \mathcal{H}_1 \otimes \dots \otimes \mathcal{H}_n $ be an $n$-partite Hilbert space. Suppose $\mathcal{N} \subset \mathcal{P}(\mathcal{H})$ is a collection of bases such that for any $\beta \in \mathcal{P}(\mathcal{H})$, there exists a $\gamma \in \mathcal{N}$ such that $\|\beta - \gamma\|_B\leq \varepsilon$. We then say that $\mathcal{N}$ is a \textbf{basis $\bm \varepsilon$-net} on $\mathcal{P}(\mathcal{H})$. 
\end{definition}

As with $\varepsilon$-nets for normalized states (see Section~\ref{sec:high-d-prob}), we can bound the size of basis $\varepsilon$-nets of product bases:

\begin{theorem}\label{thm:basis-net-thm}
    For any $\varepsilon\in (0,1)$, there exists a basis $\varepsilon$-net $\mathcal{N}$ for $\mathcal{P}((\mathbb{C}^2)^{\otimes n})$ with
    \begin{equation}
        |\mathcal{N}|\leq \left(\tfrac{5(2\sqrt{2} + 1)^2n^2}{\varepsilon^2}\right)^{8n}.
    \end{equation}
\end{theorem}

Since we may write $\mathcal{H}_A = \bigotimes_{i=1}^{N_A} \mathcal{H}_i$, where $\mathcal{H}_i$ is the Hilbert space of qubit $i\in A$, the set $\mathcal{P}(\mathcal{H}_A)$ consists of all ordered orthonormal bases that factorize as product bases over this tensor product per the definition in~\eqref{eq:p-h-def}. We may thus write the LE as
    \begin{equation} \label{eq:ltl-def-alt}
        \ltl(|\Psi\rangle) := \max_{\beta \in \mathcal P(\mathcal H_A)} \overline \tau_\beta (\Psi).
    \end{equation}
Let $\beta_{\max} \in \mathcal P(\mathcal H_A)$ be a basis such that $\overline \tau_{\beta}(\Psi) \le \overline \tau_{\beta_{\max}}(\Psi)$ for all $\beta \in \mathcal P(\mathcal H_A)$ and $K$ be defined as in Eq.~\eqref{eq:K}. By Eq.~\eqref{eq:ltl-def-alt}, $\ltl(\Psi) = \overline \tau_{\beta_{\max}}(\Psi)$. We have shown in Appendix~\ref{app:concentration-haar} that for an arbitrary $\gamma \in \mathcal P(\mathcal H_A)$, the value $\overline \tau_{\beta_{\max}}(\Psi)$ will \textit{typically} be close to $K$ as long as $\gamma$ is close enough to $\beta_{\max}$ --- i.e., the $B$-norm between $\beta_{\max}$ and $\gamma$ is small enough. Also, by Thm.~\ref{thm:basis-net-thm} there exists a basis net $\mathcal N$ such that at least one member is close to $\beta_{\max}$. Thus, we may invoke a union bound in the style of Section~\ref{sec:high-d-prob} to obtain the following theorem.

\begin{theorem} \label{thm:local-rB}
    For all $\varepsilon > 0$, the probability bound holds:
    {\small
    \begin{multline}
        \Pr_{|\Psi\rangle \sim \mu_H} \big(L^\tau_{\rm local}(\Psi) \geq K + \varepsilon \big) \\ 
        \leq 2\left(\tfrac{40(1+2\sqrt2)^2N_A^2d_A^2d_B^2}{\varepsilon^2}\right)^{8N_A} \exp\left(-\tfrac{d_Ad_B\varepsilon^2}{18\pi^3(2+4\sqrt2)^2}\right). \label{eq:local-rB}
    \end{multline}}
\end{theorem}

Combining this result with Theorem~\ref{thm:global-rB} and the fact that $\ltl(|\Psi\rangle_{AB}) \leq \ltg(|\Psi\rangle_{AB})$, we get the following corollary:

\begin{corollary} \label{cor:local-concentration}
    Let $\varepsilon, \delta>0$ be arbitrary. Then for any $d_B\geq 2$, there exists an $N_0 \in \mathbb{N}$ such that for all $d_A \geq 2^{N_0}$, we have 
    \begin{equation}
        \Pr_{|\Psi\rangle \sim \mu_H} \big( \ltl(\Psi)\geq K + \varepsilon \big) \leq \delta.
    \end{equation}
    Likewise, for any $d_A \geq 2$, there exists an $N_0' \in \mathbb{N}$ such that for all $d_B \geq 2^{N_0'}$, the above bound holds.
\end{corollary}

With Cors.~\ref{cor:global-concentration} and \ref{cor:local-concentration} in place, our final results can be summarized in Fig.~\ref{table:concentration_results}.
\begin{figure}[htb] \centering
    \begin{tabular}{|c|c|c|}
        \hline
        & {$d_A$ fixed, $d_B \to\infty$} & {$d_B$ fixed, $d_A\to\infty$} \\
        \hline
        {$\ltl$ typical value} & $< \sqrt{\frac{2}{d_B+1}}$ & $< \sqrt{\frac{2}{d_B+1}}$ \\
        \hline
        {$\ltg$ typical value} & $< \sqrt{\frac{2}{d_B+1}}$ & $> 1- \sqrt{\frac{2d_B}{d_A}}$ \\
        \hline
    \end{tabular}
    \caption{ Typical values of $\ltg$ and $\ltl$ for the two possible parameter regimes in the framework of localizable entanglement.}
    \label{table:concentration_results}
\end{figure}

\begin{figure*}
\centering
\begin{tabular}{cc}
    \includegraphics[width=0.47\linewidth]{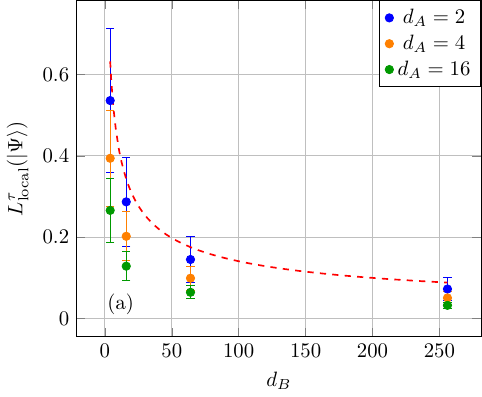}&
    \includegraphics[width=0.47\linewidth]{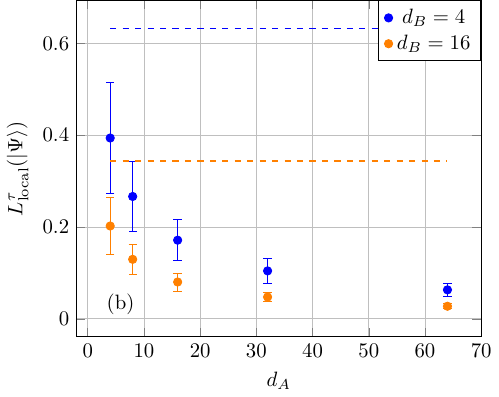}\\[2\tabcolsep]
    \includegraphics[width=0.47\linewidth]{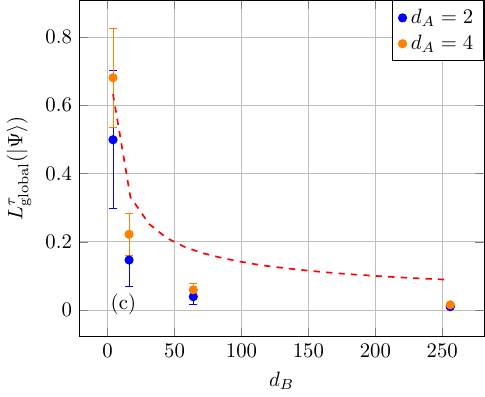}&
    \includegraphics[width=0.47\linewidth]{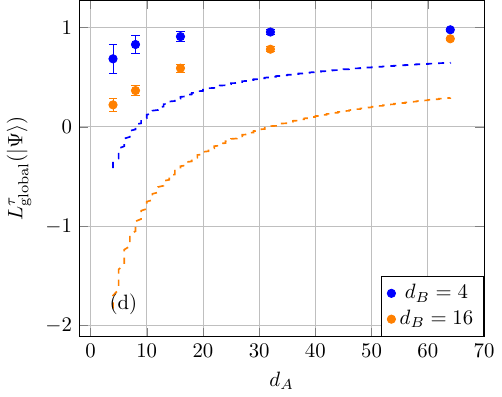}
\end{tabular}
\caption{Here we numerically visualize the scaling of $\ltg$ and $\ltl$. Equation~\eqref{eq:ltg-is-F} was used to compute the values of $\ltg$ on numerically sampled states and a COBYLA optimizer was used to compute the values of $\ltl$. Points represent numerically computed mean values of $\ltg$ and $\ltl$ across samples of 1000 Haar-random states. Error bars reflect one standard deviation for these samples. The dashed curves represent the typical bounds from Fig.~\ref{table:concentration_results}.}
\label{fig:concentration-results}
\end{figure*}

\subsubsection{Some implications}

The typical values of $\ltl$ and $\ltg$ derived in Fig.~\ref{table:concentration_results} imply that there is a near-maximal separation in the amount of entanglement we may localize on $B$ between local and global measurements whenever $d_A \gg d_B$. That is, we may localize a near-maximal value of the $n$-tangle on $B$ with global measurements on a typical state in this regime. On contrary, we may localize at most near-zero $n$-tangle on $B$ if we restrict to local measurements on a typical state. This observation implies a striking difference in the power of local versus global projective measurements. Furthermore, while it is perhaps unsurprising that we cannot localize much $n$-tangle when the size of the measured-out system $A$ is overshadowed by the size of the target system $B$, the sharp jump between near-minimal to near-maximal values of $\ltg$ when going from small $d_A$ to large $d_A$ is also notable and reminiscent of a phase transition.

Finally, we add that our concentration results have implications for our ability to extract states with maximal $n$-tangle, such as the GHZ state. Since the $n$-tangle assumes a maximum value of one, the definitions~\eqref{eq:mea-def} and~\eqref{eq:lme-def} of $\ltg$ and $\ltl$ imply that $\ltg$ and $\ltl$ upper bound the maximum probabilities with which we can extract a maximal $n$-tangle state using global and local measurements, respectively. Thus, an understanding of the typical values of $\ltg$ and $\ltl$ yields us information about the aforementioned state extraction probabilities for, \textit{e.g.}, GHZ states.

\section{Discussion and Future Directions}\label{sec:future-directions}

In this paper, we studied the extent to which multipartite entanglement, as measured by the $n$-tangle, may be localized via measurements in the limit of large system size. For this purpose, we rely on the LE and EA devised in~\cite{Vairogs2024} as benchmarks describing our ability to do so. This study focuses on describing the measure concentration phenomena of the LE and EA and their consequences for many-qubit random graph states and Haar-random states in high-dimensional spaces. 

We frame our discussion of the EA on random graph states in terms of a simple polynomial-time test introduced in~\cite{Vairogs2024} to determine whether certain graph state transformations are possible via measurements and local unitaries. This test is not always conclusive as to whether a fixed target graph state, such as a GHZ state, may be extracted from an input graph state. Thus, evaluating the probability that this test is conclusive for a typical graph state is crucial for evaluating how useful it is. We argue that the average EA over random graph states for many-qubit systems amounts to determining the probability that this test is conclusive and introduce a useful approximation for this probability (see~\eqref{eq:pr-sol-approx}). This probability decays to zero in the limit of infinite system size if the proportion of measured qubits is less than half the total system and approaches one otherwise. Thus, we expect that the test to be most useful in paring down potential resource states when we wish to measure modestly-sized subsystems to transform graph states. On the other hand, while the test is typically not conclusive as to whether or not a \textit{particular} graph state transformation is possible whenever we measure out more than half of the system, this result does guarantee that it is possible to transform a graph state into \textit{some} state of maximal $n$-tangle with global measurements on a subsystem and local unitaries. Returning to the EA picture, our result for the solution probability also implies that we can localize near-maximal values of the $n$-tangle for large systems if the share of the measured qubits forms a majority and only near-minimal values of $n$-tangle otherwise. 

One might also consider whether other useful information can be drawn from the test in the inconclusive case. Suppose we modify the test to declare that it \textit{is} possible to extract a fixed maximal $n$-tangle graph state $|G'\rangle_B$ from a source graph state $|G\rangle_{AB}$ with nonzero probability when the matrix equation \textit{does} have a solution. Suppose also that we do not modify the test's conclusion in the case that it \textit{does not} have a solution. For a random graph state, the probability that this approach incorrectly concludes that the transformation is impossible is zero. However, there are certainly examples of graph states for which a transformation into $|G'\rangle_B$ would be incorrectly declared possible. It is unclear what the probability of such an incorrect conclusion would be. In our view, it does not seem unreasonable that such a probability might simultaneously be low for large systems. We believe this scenario is worth further research. The techniques of~\cite{Ghosh2024} appear particularly useful for this purpose. 

Moving from graph states to Haar-random states, we derive explicit concentration inequalities for the LE and MEA. Due to the fact that the construction of the LE and EA involves a formidable optimization over all measurement bases, coupled with inconvenient dimensionality factors that arise through standard approaches involving $\varepsilon$-nets of \textit{states}, we introduce the notion of an $\varepsilon$-net of \textit{orthonormal bases} to derive our concentration inequalities, which may be of independent technical interest. In particular, we derive existence theorems for $\varepsilon$-nets of orthonormal bases that factorize as tensor products. Considering the limiting behavior of our concentration inequalities, we find that both the EA and LE concentrate near zero in the limit of infinite total system size if the subset of measured qubits is held constant at less than half the total system size. On the other hand, if the fraction of measured qubits is greater than half, then the EA concentrates near its maximal value of one~\cite{Vairogs2024} while the LE concentrates near zero. This asymptotic separation in the typical values of LE and EA shows that global measurements have a distinct advantage in localizing entanglement that only appears when the proportion of the measured system exceeds one-half and is not detectable by the $n$-tangle before then. 

In this direction, it would be interesting to see how our concentration results generalize when the LE and EA are defined by other entanglement measures or other localizable properties (e.g. those considered in Ref.~\cite{du2025Certifying}). The arguments used to prove our concentration results generally do not take into account detailed properties of the $n$-tangle, other than its Lipschitz continuity. Moreover, our concentration result for the LE in terms of the $n$-tangle essentially amounts to showing that the LE will concentrate near the mean value of the $n$-tangle itself. Thus, it seems plausible that our arguments could be used to show that when defined in terms of other Lipschitz-continuous entanglement measures, the LE and EA will concentrate near the mean values of the entanglement measure. However, our proof that the typical values of the LE and EA for the $n$-tangle separate relies on the concentration result for the EA from~\cite{Vairogs2024}, which uses an analytic expression unique to the $n$-tangle. Thus, it is unclear for what other entanglement measures we might observe a separation between local and global measurements in the amount of entanglement we may localize.

\begin{acknowledgments}
C.V. is supported by NSF grant No. 2137953. J.L.B is supported by a National Science Foundation Mathematical Sciences Postdoctoral Research Fellowship under Award No.~2402287 as well as an IQUIST Postdoctoral Fellowship. C.V. would like to acknowledge helpful discussions with Daniel Belkin, Jens Eisert, and You Zhou.
\end{acknowledgments}

%%%%%%%%%%%%%%%%%%%%%%%
% References
\clearpage
\onecolumngrid
\bibliography{main}

\setcounter{section}{0}
\setcounter{figure}{0}
\renewcommand{\figurename}{Sup. Fig.}

%%%%%%%%%%%%%%%%%%%%%%%
% Appendix
\appendix
\section*{Appendix}
\maketitle

\section{Derivation of $p_s(\Eunif)$ estimate}\label{app:gs-concentration}

The goal for this section is to prove Theorem~\ref{thm:solution-prob}. Our basic tool will be the following Lemma. Recall that for arbitrary graph state $|G\rangle_{AB}$, we define $G_B \coloneqq \tr_A[|G\rangle\langle G|_{AB}]$.

\begin{lemma}\label{lemma:lt-is-tr-quot}
    Let $|G\rangle_{AB}$ be an arbitrary graph state. Then
    \begin{equation}
        L^\tau(|G\rangle_{AB}) = \frac{\tr[G_B \tilde{G}_B]}{\tr[G_B^2]}.
    \end{equation}
\end{lemma}
\begin{proof}
    There are two cases:

    \noindent{\textbf{Case 1:}} The matrix equation $\Bgammaba \mathbf{x} = \D$ has a solution. By Theorem~\ref{thm:ltg-in-terms-of-matrices}, we have $L^\tau(|G\rangle_{AB}) = 1$. Equation~\eqref{eq:ltg-is-F} then implies that $F(G_B, \tilde{G}_B) = 1$, which implies that $G_B = \tilde{G}_B$. Hence, 
    \begin{equation}
        \frac{\tr[G_B\tilde{G}_B]}{\tr[G_B^2]} = 1 = L^\tau(|G\rangle_{AB}).
    \end{equation}

    \noindent{\textbf{Case 2:}} The matrix equation $\Bgammaba \mathbf{x} = \D$ has no solution. By Theorem~\ref{thm:ltg-in-terms-of-matrices}, we have $L^\tau(|G\rangle_{AB}) = 0$. Equation~\eqref{eq:ltg-is-F} then implies that $F(G_B, \tilde{G}_B) = 0$. Therefore, $G_B$ and $\tilde{G}_B$ have orthogonal support~\cite{Wilde2016}[Eq. 9.117]. Hence, $\tr[G_B \tilde{G}_B] = 0$, meaning that 
    \begin{equation}
        \frac{\tr[G_B\tilde{G}_B]}{\tr[G_B^2]} = 0 = L^\tau(|G\rangle_{AB}).
    \end{equation}

    Since equality holds in both cases, the statement of the lemma follows. 
\end{proof}

Combining Lemma~\ref{lemma:lt-is-tr-quot} with~\eqref{eq:prob-ltg-equiv}, it follows that
\begin{equation}\label{eq:ps-workable-expression}
    p_s(\Eunif) = \mathbb{E}_{|G\rangle \sim \Eunif}\left[ \frac{\tr[G_B \tilde{G}_B]}{\tr[G_B^2]}\right].
\end{equation}
We will ultimately use $\mathbb{E}[\tr G_B^2]$ and $\var[\tr G_B^2]$ from \citep[Thm.~1]{zhou2022} and \citep[Thm.~3]{zhou2022} respectively to construct upper and lower bounds on $\mathbb{E}\left[\frac{\tr[G_B\tilde{G}_B]}{\tr G_B^2}\right]$ using Chebyshev's inequality. To this end, we first need to compute $\mathbb{E}[\tr[G_B\tilde{G}_B]]$. The following two lemmas assist us in this calculation. \\

Let us introduce some notation. For $\mathbf{z} \in \mathbb{Z}_2^N$, denote by $\overline{\mathbf{z}} \in \mathbb{Z}_2^N$ the bitstring obtained by inverting each bit of $\mathbf{z}$, i.e. $\mathbf{z} \oplus \overline{\mathbf{z}} = \mathbf{1}$, where $\oplus$ is the bitwise xor/binary addition operator. Let $h(\mathbf{z})$ denote the Hamming weight (number of 1 bits) in $\mathbf{z}$. Let us label each qubit pair of the $N$ qubits by $e_1, \dots, e_M$ for $M = N(N- 1)/2$. For qubit pair $e_k = (i,j)$, let $\mathcal{E}_{e_k}$ denote the distribution over the set $\{I_{e_k}, \mathrm{CZ}_{e_k}\}$ with each operator occuring with probability $1.2$

\begin{lemma}\label{lemma:i-j-diff-set}
Suppose that $e_k = (i,j), i \in A, j \in B$. For any product $P$ of $CZ$ operators acting on arbitrary pairs o qubits that are not $(i,j)$, we have
\begin{align}
    \begin{aligned}
    \mathbb{E}_{U_{k} \sim \mathcal{E}_{e_{k}}} & {\left[\tr\left[\left(S_{A} \otimes Y^{\otimes 2}\right) P^{\otimes 2} U_{k}^{\otimes 2}\left(\left|\mathbf{z}_{1}\right\rangle\left\langle\mathbf{z}_{2}\right| \otimes\left|\mathbf{z}_{3}\right\rangle\left\langle\mathbf{z}_{4}\right|\right) U_{k}^{\otimes 2\dagger} P^{\otimes 2 \dagger}\right]\right] } \\
    & =F_{i j}\left(\mathbf{z}_{1}, \mathbf{z}_{2}, \mathbf{z}_{3}, \mathbf{z}_{4}\right) \tr\left[\left(S_{A} \otimes Y^{\otimes 2}\right) P^{\otimes 2}\left(\left|\mathbf{z}_{1}\right\rangle\left\langle\mathbf{z}_{2}\right| \otimes\left|\mathbf{z}_{3}\right\rangle\left\langle\mathbf{z}_{4}\right|\right) P^{\otimes 2 \dagger}\right]
    \end{aligned}
    \end{align}
    where 
    \begin{align}
    F_{i j}\left(\mathbf{z}_{1}, \mathbf{z}_{2}, \mathbf{z}_{3}, \mathbf{z}_{4}\right) = 
    \begin{cases}
        0 \text{ if } (\mathbf{z}_{1})_{i} \oplus (\mathbf{z}_{2})_{i}=1 \textnormal{ and } \left(\mathbf{z}_{1}\right)_{j} \oplus\left(\mathbf{z}_{3}\right)_{j} = 0\\
        1 \text{ otherwise}
    \end{cases}
    \end{align}
\end{lemma}
\begin{proof}
Suppose \(U_{k}=I_{e_{k}}\). It is clear that the argument of the expectation value becomes
\begin{align}
\tr\left[\left(S_{A} \otimes Y^{\otimes 2}\right) P\left(\left|\mathbf{z}_{1}\right\rangle\left\langle\mathbf{z}_{2}\right| \otimes\left|\mathbf{z}_{3}\right\rangle\left\langle\mathbf{z}_{4}\right|\right) P^{\dagger}\right]
\end{align}
in this case.

Now suppose \(U_{k}=C Z_{e_{k}}\). First observe that on a single qubit $|x\rangle$, $\sigma_y|x\rangle = (-1)^xi|x\oplus1\rangle$. Thus, for a quantum state represented by bitstring $s$, $\sigma_y^{\otimes |s|}|x\rangle = (-1)^{h(x)}i^{|s|}|x\oplus\mathbf{1}\rangle$. Also, for any two qubits in a computational basis state, a $CZ$ gate only introduces a $\pm 1$ phase (more precisely, $CZ|i\rangle|j\rangle = (-1)^{i\wedge j}|i\rangle|ij\rangle$ for $i,j \in \{0,1\}$). Thus, we may write 
\begin{equation}
    P^{\otimes 2}(|\mathbf{z}_1\rangle\langle\mathbf{z}_2|\otimes|\mathbf{z}_3\rangle\langle \mathbf{z}_4|)P^{\otimes 2 \dagger} = c |\mathbf{z}_1\rangle\langle\mathbf{z}_2|\otimes|\mathbf{z}_3\rangle\langle \mathbf{z}_4|
\end{equation}
for some constant $c = \pm 1$. Then the argument of the expectation value assumes a value of
\begin{align}
\begin{aligned}
\tr & {\left[\left(S_{A} \otimes Y^{\otimes 2}\right) P^{\otimes 2} C Z_{(i, j)}^{\otimes 2}\left(\left|\mathbf{z}_{1}\right\rangle\left\langle\mathbf{z}_{2}\right| \otimes\left|\mathbf{z}_{3}\right\rangle\left\langle\mathbf{z}_{4}\right|\right) C Z_{(i, j)}^{\otimes 2 \dagger} P^{\otimes 2 \dagger}\right] } \\
& =(-1)^{\left(\mathbf{z}_{1}\right)_{i} \wedge\left(\mathbf{z}_{1}\right)_{j}}(-1)^{\left(\mathbf{z}_{3}\right)_{i} \wedge\left(\mathbf{z}_{3}\right)_{j}}(-1)^{\left(\mathbf{z}_{2}\right)_{i} \wedge\left(\mathbf{z}_{2}\right)_{j}}(-1)^{\left(\mathbf{z}_{4}\right)_{i} \wedge\left(\mathbf{z}_{4}\right)_{j}} \tr\left[\left(S_{A} \otimes Y^{\otimes 2}\right) P^{\otimes 2}\left(\left|\mathbf{z}_{1}\right\rangle\left\langle\mathbf{z}_{2}\right| \otimes\left|\mathbf{z}_{3}\right\rangle\left\langle\mathbf{z}_{4}\right|\right) P^{\otimes 2 \dagger}\right] \\
& =(-1)^{\left(\mathbf{z}_{1}\right)_{i} \wedge\left(\mathbf{z}_{1}\right)_{j}}(-1)^{\left(\mathbf{z}_{3}\right)_{i} \wedge\left(\mathbf{z}_{3}\right)_{j}}(-1)^{\left(\mathbf{z}_{2}\right)_{i} \wedge\left(\mathbf{z}_{2}\right)_{j}}(-1)^{\left(\mathbf{z}_{4}\right)_{i} \wedge\left(\mathbf{z}_{4}\right)_{j}} c \tr\left[\left(S_{A} \otimes Y^{\otimes 2}\right)\left(\left|\mathbf{z}_{1}\right\rangle\left\langle\mathbf{z}_{2}\right| \otimes\left|\mathbf{z}_{3}\right\rangle\left\langle\mathbf{z}_{4}\right|\right)\right] \\
& =(-1)^{\left(\mathbf{z}_{1}\right)_{i} \wedge\left(\mathbf{z}_{1}\right)_{j}}(-1)^{\left(\mathbf{z}_{3}\right)_{i} \wedge\left(\mathbf{z}_{3}\right)_{j}}(-1)^{\left(\mathbf{z}_{2}\right)_{i} \wedge\left(\mathbf{z}_{2}\right)_{j}}(-1)^{\left(\mathbf{z}_{4}\right)_{i} \wedge\left(\mathbf{z}_{4}\right)_{j}} c \tr\left[\left|\mathbf{z}_{3, A} Y\mathbf{z}_{1, B}\right\rangle\left\langle\mathbf{z}_{2, A} \mathbf{z}_{2, B}\right| \otimes\left|\mathbf{z}_{1, A} Y\mathbf{z}_{3, B}\right\rangle\left\langle\mathbf{z}_{4, A} \mathbf{z}_{4, B}\right|\right] \\
& =(-1)^{\left(\mathbf{z}_{1}\right)_{i} \wedge\left(\mathbf{z}_{1}\right)_{j}}(-1)^{\left(\mathbf{z}_{3}\right)_{i} \wedge\left(\mathbf{z}_{3}\right)_{j}}(-1)^{\left(\mathbf{z}_{2}\right)_{i} \wedge\left(\mathbf{z}_{2}\right)_{j}}(-1)^{\left(\mathbf{z}_{4}\right)_{i} \wedge\left(\mathbf{z}_{4}\right)_{j}} c \\
& \times (-1)^{h(z_{1B})+h(z_{3B})} i^{2N_B} \delta_{\mathbf{z}_{3, A}, \mathbf{z}_{2, A}} \delta_{\overline{\mathbf{z}_{1, B}}, \mathbf{z}_{2, B}} \delta_{\mathbf{z}_{1, A}, \mathbf{z}_{4, A}} \delta_{\overline{\mathbf{z}_{3, B}}, \mathbf{z}_{4, B}} \\
& =(-1)^{\left(\mathbf{z}_{1}\right)_{i} \wedge\left(\mathbf{z}_{1}\right)_{j}}(-1)^{\left(\mathbf{z}_{2}\right)_{i} \wedge\left(\mathbf{z}_{3}\right)_{j}}(-1)^{\left(\mathbf{z}_{2}\right)_{i} \wedge \overline{\left(\mathbf{z}_{1}\right)_{j}}}(-1)^{\left(\mathbf{z}_{1}\right)_{i} \wedge \overline{\left(\mathbf{z}_{3}\right)_{j}}} c \\
& \times (-1)^{h(z_{1B})+h(z_{3B})} i^{2N_B} \delta_{\mathbf{z}_{3, A}, \mathbf{z}_{2, A}} \delta_{\overline{\mathbf{z}_{1, B}}, \mathbf{z}_{2, B}} \delta_{\mathbf{z}_{1, A}, \mathbf{z}_{4, A}} \delta_{\overline{\mathbf{z}_{3, B}}, \mathbf{z}_{4, B}} \\
& =(-1)^{\left(\mathbf{z}_{1}\right)_{i} \wedge\left(\mathbf{z}_{1}\right)_{j}}(-1)^{\left(\mathbf{z}_{2}\right)_{i} \wedge\left(\mathbf{z}_{3}\right)_{j}}(-1)^{\left(\mathbf{z}_{2}\right)_{i} \wedge \overline{\left(\mathbf{z}_{1}\right)_{j}}}(-1)^{\left(\mathbf{z}_{1}\right)_{i} \wedge \overline{\left(\mathbf{z}_{3}\right)_{j}}} \tr\left[\left(S_{A} \otimes Y^{\otimes 2}\right) P^{\otimes 2}\left(\left|\mathbf{z}_{1}\right\rangle\left\langle\mathbf{z}_{2}\right| \otimes\left|\mathbf{z}_{3}\right\rangle\left\langle\mathbf{z}_{4}\right|\right) P^{\otimes 2 \dagger}\right]
\end{aligned}
\end{align}
By averaging the two values of the expectation value argument when \(U_{k}=I_{e_{k}}\) and \(U_{k}=C Z_{e_{k}}\), we get that

\begin{align}
\begin{aligned}
\mathbb{E}_{U_{k} \sim \mathcal{E}_{e_{k}}} & {\left[\tr\left[\left(S_{A} \otimes Y^{\otimes 2}\right) P^{\otimes 2} U_{k}^{\otimes 2}\left(\left|\mathbf{z}_{1}\right\rangle\left\langle\mathbf{z}_{2}\right| \otimes\left|\mathbf{z}_{3}\right\rangle\left\langle\mathbf{z}_{4}\right|\right) U_{k}^{\otimes 2 \dagger} P^{\otimes 2 \dagger}\right]\right] } \\
& =\frac{1}{2}\left(1+\phi\left(\mathbf{z}_{1}, \mathbf{z}_{2}, \mathbf{z}_{3}, \mathbf{z}_{4}\right)\right) \tr\left[\left(S_{A} \otimes Y^{\otimes 2}\right) P^{\otimes 2}\left(\left|\mathbf{z}_{1}\right\rangle\left\langle\mathbf{z}_{2}\right| \otimes\left|\mathbf{z}_{3}\right\rangle\left\langle\mathbf{z}_{4}\right|\right) P^{\otimes 2 \dagger}\right]
\end{aligned}
\end{align}
where
\begin{align}
\phi_{i j}\left(\mathbf{z}_{1}, \mathbf{z}_{2}, \mathbf{z}_{3}, \mathbf{z}_{4}\right)=(-1)^{\left(\mathbf{z}_{1}\right)_{i} \wedge\left(\mathbf{z}_{1}\right)_{j}}(-1)^{\left(\mathbf{z}_{2}\right)_{i} \wedge\left(\mathbf{z}_{3}\right)_{j}}(-1)^{\left(\mathbf{z}_{2}\right)_{i} \wedge \overline{\left(\mathbf{z}_{1}\right)_{j}}}(-1)^{\left(\mathbf{z}_{1}\right)_{i} \wedge \overline{\left(\mathbf{z}_{3}\right)_{j}}}
\end{align}

One may carefully work out the following table for the values of \(\left(1+\phi_{i j}\right) / 2\) :
\begin{center}
\begin{tabular}{ccccc}
 & 00 & 01 & 10 & 11 \\
00 & 1 & 1 & 0 & 1 \\
01 & 1 & 1 & 1 & 0 \\
10 & 0 & 1 & 1 & 1 \\
11 & 1 & 0 & 1 & 1 \\
\end{tabular}
\end{center}

Here, the columns are labeled by values of \(\left(\mathbf{z}_{1}\right)_{i},\left(\mathbf{z}_{1}\right)_{j}\) and the rows are labeled by values of \(\left(\mathbf{z}_{2}\right)_{i},\left(\mathbf{z}_{3}\right)_{j}\). It follows that \(F_{i j}\left(\mathbf{z}_{1}, \mathbf{z}_{2}, \mathbf{z}_{3}, \mathbf{z}_{4}\right) = \left(1+\phi_{i j}\right) / 2\) is zero if \((\mathbf{z}_{1})_{i} \oplus (\mathbf{z}_{2})_{i}=1 \wedge \left(\mathbf{z}_{1}\right)_{j} \oplus\left(\mathbf{z}_{3}\right)_{j} = 0\); otherwise, it is one. 
\end{proof}

\begin{lemma}\label{lemma:i-j-same-set}
Suppose that $e_k = (i,j), i,j \in B$.  For any product $P$ of $CZ$ operators acting on arbitrary pairs o qubits that are not $(i,j)$, we have
    \begin{align}
    \begin{aligned}
    \mathbb{E}_{U_{k} \sim \mathcal{E}_{e_{k}}} & {\left[\tr\left[\left(S_{A} \otimes Y^{\otimes 2}\right) P^{\otimes 2} U_{k}^{\otimes 2}\left(\left|\mathbf{z}_{1}\right\rangle\left\langle\mathbf{z}_{2}\right| \otimes\left|\mathbf{z}_{3}\right\rangle\left\langle\mathbf{z}_{4}\right|\right) U_{k}^{\otimes 2\dagger} P^{\otimes 2 \dagger}\right]\right] } \\
    & =H_{i j}\left(\mathbf{z}_{1}, \mathbf{z}_{2}, \mathbf{z}_{3}, \mathbf{z}_{4}\right) \tr\left[\left(S_{A} \otimes Y^{\otimes 2}\right) P^{\otimes 2}\left(\left|\mathbf{z}_{1}\right\rangle\left\langle\mathbf{z}_{2}\right| \otimes\left|\mathbf{z}_{3}\right\rangle\left\langle\mathbf{z}_{4}\right|\right) P^{\otimes 2 \dagger}\right]
    \end{aligned}
    \end{align}
    where 
    \begin{align}
    H_{i j}\left(\mathbf{z}_{1}, \mathbf{z}_{2}, \mathbf{z}_{3}, \mathbf{z}_{4}\right) = 
    \begin{cases}
        1 \text{ if } (\mathbf{z}_{1})_{i} \oplus (\mathbf{z}_1)_{j} = \left(\mathbf{z}_{3}\right)_{i} \oplus\left(\mathbf{z}_{3}\right)_{j}\\
        0 \text{ otherwise}
    \end{cases}
    \end{align}
\end{lemma}
\begin{proof}
The proof is very similar to the proof of the preceding lemma. 
\begin{align}
\begin{aligned}
\tr & {\left[\left(S_{A} \otimes Y^{\otimes 2}\right) P^{\otimes 2} C Z_{(i, j)}^{\otimes 2}\left(\left|\mathbf{z}_{1}\right\rangle\left\langle\mathbf{z}_{2}\right| \otimes\left|\mathbf{z}_{3}\right\rangle\left\langle\mathbf{z}_{4}\right|\right) C Z_{(i, j)}^{\otimes 2 \dagger} P^{\otimes 2 \dagger}\right] } \\
& =(-1)^{\left(\mathbf{z}_{1}\right)_{i} \wedge\left(\mathbf{z}_{1}\right)_{j}}(-1)^{\left(\mathbf{z}_{3}\right)_{i} \wedge\left(\mathbf{z}_{3}\right)_{j}}(-1)^{\left(\mathbf{z}_{2}\right)_{i} \wedge\left(\mathbf{z}_{2}\right)_{j}}(-1)^{\left(\mathbf{z}_{4}\right)_{i} \wedge\left(\mathbf{z}_{4}\right)_{j}} \tr\left[\left(S_{A} \otimes Y^{\otimes 2}\right) P^{\otimes 2}\left(\left|\mathbf{z}_{1}\right\rangle\left\langle\mathbf{z}_{2}\right| \otimes\left|\mathbf{z}_{3}\right\rangle\left\langle\mathbf{z}_{4}\right|\right) P^{\otimes 2 \dagger}\right] \\
& =(-1)^{\left(\mathbf{z}_{1}\right)_{i} \wedge\left(\mathbf{z}_{1}\right)_{j}}(-1)^{\left(\mathbf{z}_{3}\right)_{i} \wedge\left(\mathbf{z}_{3}\right)_{j}}(-1)^{\left(\mathbf{z}_{2}\right)_{i} \wedge\left(\mathbf{z}_{2}\right)_{j}}(-1)^{\left(\mathbf{z}_{4}\right)_{i} \wedge\left(\mathbf{z}_{4}\right)_{j}} c \tr\left[\left(S_{A} \otimes Y^{\otimes 2}\right)\left(\left|\mathbf{z}_{1}\right\rangle\left\langle\mathbf{z}_{2}\right| \otimes\left|\mathbf{z}_{3}\right\rangle\left\langle\mathbf{z}_{4}\right|\right)\right] \\
& =(-1)^{\left(\mathbf{z}_{1}\right)_{i} \wedge\left(\mathbf{z}_{1}\right)_{j}}(-1)^{\left(\mathbf{z}_{3}\right)_{i} \wedge\left(\mathbf{z}_{3}\right)_{j}}(-1)^{\left(\mathbf{z}_{2}\right)_{i} \wedge\left(\mathbf{z}_{2}\right)_{j}}(-1)^{\left(\mathbf{z}_{4}\right)_{i} \wedge\left(\mathbf{z}_{4}\right)_{j}} c \tr\left[\left|\mathbf{z}_{3, A} Y\mathbf{z}_{1, B}\right\rangle\left\langle\mathbf{z}_{2, A} \mathbf{z}_{2, B}\right| \otimes\left|\mathbf{z}_{1, A} Y\mathbf{z}_{3, B}\right\rangle\left\langle\mathbf{z}_{4, A} \mathbf{z}_{4, B}\right|\right] \\
& =(-1)^{\left(\mathbf{z}_{1}\right)_{i} \wedge\left(\mathbf{z}_{1}\right)_{j}}(-1)^{\left(\mathbf{z}_{3}\right)_{i} \wedge\left(\mathbf{z}_{3}\right)_{j}}(-1)^{\left(\mathbf{z}_{2}\right)_{i} \wedge\left(\mathbf{z}_{2}\right)_{j}}(-1)^{\left(\mathbf{z}_{4}\right)_{i} \wedge\left(\mathbf{z}_{4}\right)_{j}} c \\
&\times (-1)^{h(z_{1B})+h(z_{3B})} i^{2N_B} \delta_{\mathbf{z}_{3, A}, \mathbf{z}_{2, A}} \delta_{\overline{\mathbf{z}_{1, B}}, \mathbf{z}_{2, B}} \delta_{\mathbf{z}_{1, A}, \mathbf{z}_{4, A}} \delta_{\overline{\mathbf{z}_{3, B}}, \mathbf{z}_{4, B}} \\
& =(-1)^{\left(\mathbf{z}_{1}\right)_{i} \wedge\left(\mathbf{z}_{1}\right)_{j}}(-1)^{\left(\mathbf{z}_{3}\right)_{i} \wedge\left(\mathbf{z}_{3}\right)_{j}}(-1)^{\overline{\left(\mathbf{z}_{1}\right)_{i}} \wedge \overline{\left(\mathbf{z}_{1}\right)_{j}}} (-1)^{\overline{\left(\mathbf{z}_{3}\right)_{i}} \wedge \overline{\left(\mathbf{z}_{3}\right)_{j}}} c \\
&\times (-1)^{h(z_{1B})+h(z_{3B})} i^{2N_B} \delta_{\mathbf{z}_{3, A}, \mathbf{z}_{2, A}} \delta_{\overline{\mathbf{z}_{1, B}}, \mathbf{z}_{2, B}} \delta_{\mathbf{z}_{1, A}, \mathbf{z}_{4, A}} \delta_{\overline{\mathbf{z}_{3, B}}, \mathbf{z}_{4, B}} \\
& =(-1)^{\left(\mathbf{z}_{1}\right)_{i} \wedge\left(\mathbf{z}_{1}\right)_{j}}(-1)^{\left(\mathbf{z}_{3}\right)_{i} \wedge\left(\mathbf{z}_{3}\right)_{j}}(-1)^{\overline{\left(\mathbf{z}_{1}\right)_{i}} \wedge \overline{\left(\mathbf{z}_{1}\right)_{j}}} (-1)^{\overline{\left(\mathbf{z}_{3}\right)_{i}} \wedge \overline{\left(\mathbf{z}_{3}\right)_{j}}} \tr\left[\left(S_{A} \otimes Y^{\otimes 2}\right) P^{\otimes 2}\left(\left|\mathbf{z}_{1}\right\rangle\left\langle\mathbf{z}_{2}\right| \otimes\left|\mathbf{z}_{3}\right\rangle\left\langle\mathbf{z}_{4}\right|\right) P^{\otimes 2 \dagger}\right]
\end{aligned}
\end{align}
By averaging the two values of the expectation value argument when \(U_{k}=I_{e_{k}}\) and \(U_{k}=C Z_{e_{k}}\), we get that

\begin{align}
\begin{aligned}
\mathbb{E}_{U_{k} \sim \mathcal{E}_{e_{k}}} & {\left[\tr\left[\left(S_{A} \otimes Y^{\otimes 2}\right) P^{\otimes 2} U_{k}^{\otimes 2}\left(\left|\mathbf{z}_{1}\right\rangle\left\langle\mathbf{z}_{2}\right| \otimes\left|\mathbf{z}_{3}\right\rangle\left\langle\mathbf{z}_{4}\right|\right) U_{k}^{\otimes 2 \dagger} P^{\otimes 2 \dagger}\right]\right] } \\
& =\frac{1}{2}\left(1+\phi\left(\mathbf{z}_{1}, \mathbf{z}_{2}, \mathbf{z}_{3}, \mathbf{z}_{4}\right)\right) \tr\left[\left(S_{A} \otimes Y^{\otimes 2}\right) P^{\otimes 2}\left(\left|\mathbf{z}_{1}\right\rangle\left\langle\mathbf{z}_{2}\right| \otimes\left|\mathbf{z}_{3}\right\rangle\left\langle\mathbf{z}_{4}\right|\right) P^{\otimes 2 \dagger}\right]
\end{aligned}
\end{align}
where
\begin{align}
\phi_{i j}\left(\mathbf{z}_{1}, \mathbf{z}_{2}, \mathbf{z}_{3}, \mathbf{z}_{4}\right)=(-1)^{\left(\mathbf{z}_{1}\right)_{i} \wedge\left(\mathbf{z}_{1}\right)_{j}}(-1)^{\left(\mathbf{z}_{3}\right)_{i} \wedge\left(\mathbf{z}_{3}\right)_{j}}(-1)^{\overline{\left(\mathbf{z}_{1}\right)_{i}} \wedge \overline{\left(\mathbf{z}_{1}\right)_{j}}}(-1)^{\overline{\left(\mathbf{z}_{3}\right)_{i}} \wedge \overline{\left(\mathbf{z}_{3}\right)_{j}}}
\end{align}

One may carefully work out the following table for the values of \(\left(1+\phi_{i j}\right) / 2\) :
\begin{center}
\begin{tabular}{ccccc}
 & 00 & 01 & 10 & 11 \\
00 & 0 & 1 & 1 & 0 \\
01 & 1 & 0 & 0 & 1 \\
10 & 1 & 0 & 0 & 1 \\
11 & 0 & 1 & 1 & 0 \\
\end{tabular}
\end{center}

Here, the columns are labeled by values of \(\left(\mathbf{z}_{1}\right)_{i},\left(\mathbf{z}_{1}\right)_{j}\) and the rows are labeled by values of \(\left(\mathbf{z}_{3}\right)_{i},\left(\mathbf{z}_{3}\right)_{j}\). It follows that \(H_{i j}\left(\mathbf{z}_{1}, \mathbf{z}_{2}, \mathbf{z}_{3}, \mathbf{z}_{4}\right) = \left(1+\phi_{i j}\right) / 2\) is one if \((\mathbf{z}_{1})_{i} \oplus (\mathbf{z}_1)_{j} = \left(\mathbf{z}_{3}\right)_{i} \oplus\left(\mathbf{z}_{3}\right)_{j}\); otherwise, it is 0.
\end{proof}

Now we are ready to derive a expression for $\mathbb{E}[\tr[G_B\tilde{G}_B]]$, as follows: 

\begin{lemma} \label{lem:approximation}
\begin{align}
\mathbb{E}_{|G\rangle_{AB} \sim \Eunif}\left[\tr\left(G_B\tilde{G}_B\right)\right] = \frac{d_{A}+1}{d_{A} d_{B}}.
\end{align}
\end{lemma}
\begin{proof}
We compute $\mathbb{E}(\tr[G_B\tilde{G}_B])$ using the same approach as $\mathbb{E}(\tr[G_B^2])$. Observe $\tilde{G}_B = \sigma_y^{\otimes N_B} G_B \sigma_y^{\otimes N_B}$ (because $\sigma_y$ is Hermitian). Denote $Y:= \sigma_y^{\otimes N_B}$. Then  $\mathbb{E}(\tr[G_B\tilde{G}_B]) = \mathbb{E}(\tr[G_BYG_BY])$. \\

Note that for any function \(F\) over quantum states, we have
\begin{align}
\begin{aligned}
\mathbb{E}_{|G\rangle \sim \Eunif}[F(|G\rangle\langle G|)] & =\frac{1}{2^{N(N-1) / 2}} \sum_{i_{1}=0}^{1} \cdots \sum_{i_{M}=0}^{1} F\left(C Z_{e_{1}}^{i_{1}} \ldots C Z_{e_{M}}^{i_{M}}(|+\rangle\langle+|)^{\otimes N}\left(C Z_{e_{M}}^{i_{M}}\right)^{\dagger} \ldots\left(C Z_{e_{1}}\right)^{\dagger}\right) \\
& =\mathbb{E}_{U_{1} \sim \mathcal{E}_{e_{1}}}\left[\ldots \mathbb{E}_{U_M \sim \mathcal{E}_{e_{M}}}\left[F\left(U_{1} \ldots U_{K}(|+\rangle\langle+|)^{\otimes N} U_{M}^{\dagger} \ldots U_{1}^{\dagger}\right)\right] \ldots\right]
\end{aligned}
\end{align}
where \(\mathcal{E}_{e_{i}}\) denotes the distribution over the the set \(\left\{I_{e_{i}}, C Z_{e_{M}}\right\}\), with each operator occurring with probability \(1 / 2\). Essentially, we can break the expectation value over all graph states into a nested expectation value of $\mathcal{E}_{e_k}$ over all qubit pairs $e_k$.\\

For any state $|\Psi\rangle \in \mathcal{H}_{A} \otimes \mathcal{H}_{B}$ and any unitary $U_{A}$ respectively acting on solely the $A$ subsystem, we have
\begin{align}
\tr_A[|\Psi\rangle\langle\Psi|] = \tr_A[(U_A\otimes I)|\Psi\rangle\langle\Psi|(U_A\otimes I)^\dagger]
\end{align}
Thus, for any constant unitary $U_A$ over the $A$ subsystem, we have
\begin{align}
\mathbb{E}_{|G\rangle \sim \Eunif}[\tr(G_B Y G_B Y)] = \mathbb{E}_{|G\rangle \sim \Eunif}[\tr(G_B' Y G_B' Y)], \text{where } G_B' = \tr_A[(U_A \otimes I) |G\rangle\langle G| (U_A \otimes I)^\dagger]
\end{align}
Assuming without loss of generality that the qubit pairs consisting of both qubits in \(A\) make up the last \(M-K\) of the pairs \(e_{1}, \ldots, e_{M}\), we may consequently write that
\begin{align}
\begin{aligned}
\mathbb{E}_{|G\rangle \sim \Eunif}\left[\tr\left[G_{B}YG_BY\right]\right] & =\mathbb{E}_{U_{1} \sim \mathcal{E}_{e_{1}}}\left[\ldots \mathbb{E}_{U_{M} \sim \mathcal{E}_{e_{M}}}\left[\tr\left[\left[\tr_A\left[U_{1} \ldots U_{M}\left((|+\rangle\langle+|)^{\otimes N}\right) U_{M}^\dagger \ldots U_{1}^\dagger\right]\cdot Y\right]^{2}\right]\right]\right] \\
& =\mathbb{E}_{U_{1} \sim \mathcal{E}_{e_{1}}}\left[\ldots \mathbb{E}_{U_{K} \sim \mathcal{E}_{e_{K}}}\left[\tr\left[\left[\tr_A\left[U_{1} \ldots U_{K}\left((|+\rangle\langle+|)^{\otimes N}\right) U_{K}^\dagger \ldots U_{1}^\dagger\right]\cdot Y\right]^{2}\right]\right]\right]
\end{aligned}
\end{align}
By the swap trick, we have for arbitrary state $|\Psi\rangle_{AB} \in \mathcal{H}_A \otimes \mathcal{H}_B$ with marginal $\Psi_B \coloneqq \tr_A[|\Psi\rangle\langle \Psi|]$ that
\begin{align}
\tr\left[\Psi_BY\Psi_BY\right]=\tr\left[\left(S_{A} \otimes Y^{\otimes 2}\right) \Psi^{\otimes 2}\right]
\end{align}
where \(S_{A}\) is the swap operator acting on just the \(A\) subsystems of \(\left(\mathcal{H}_{A} \otimes \mathcal{H}_{B}\right)^{\otimes 2}\). Then
\begin{align}
\begin{aligned}
& \mathbb{E}_{|G\rangle \sim \Eunif}\left[\tr\left[G_{B}YG_BY\right]\right] \\
& = \mathbb{E}_{U_{1} \sim \mathcal{E}_{e_{1}}}\left[\ldots \mathbb{E}_{U_{K} \sim \mathcal{E}_{e_{K}}}\left[\tr\left[\left(S_{A} \otimes Y^{\otimes 2}\right) U_{1}^{\otimes 2} \ldots U_{K}^{\otimes 2}\left((|+\rangle\langle+|)^{\otimes N}\right)^{\otimes 2} U_{K}^{\dagger \otimes 2} \ldots U_{1}^{\dagger \otimes 2}\right]\right] \ldots \right]
\end{aligned}
\end{align}
Since
\begin{align}
(|+\rangle\langle+|)^{\otimes N}=\frac{1}{2^N} \sum_{\mathbf{z}, \mathbf{z}' \in \mathbb{Z}_{2}^{N}}|\mathbf{z}\rangle\langle\mathbf{z}'|
\end{align}
it follows that
\begin{align}
\begin{aligned}
& \mathbb{E}_{|G\rangle \sim \Eunif}\left[\tr\left[G_{B}YG_BY\right]\right] \\
& =\frac{1}{2^{2 N}} \sum_{\mathbf{z}_{1}, \mathbf{z}_{2}, \mathbf{z}_{3}, \mathbf{z}_{4} \in \mathbb{Z}_{2}^{N}} \mathbb{E}_{U_{1} \sim \mathcal{E}_{e_{1}}}\left[\ldots \mathbb{E}_{U_{K} \sim \mathcal{E}_{e_{K}}}\left[\tr\left[\left(S_{A} \otimes Y^{\otimes 2}\right) U_{1}^{\otimes 2} \ldots U_{K}^{\otimes 2}\left(\left|\mathbf{z}_{1}\right\rangle\left\langle\mathbf{z}_{2}\right| \otimes\left|\mathbf{z}_{3}\right\rangle\left\langle\mathbf{z}_{4}\right|\right) U_{K}^{\dagger \otimes 2} \ldots U_{1}^{\dagger \otimes 2}\right]\right] \ldots\right] \label{eq:trace-bit-string-expansion}
\end{aligned}
\end{align}

There are two possibilities: in each qubit pair, either one is in $A$ and one is in $B$ or both are in $B$. Suppose that the qubit pairs of the former case make up the first $L$ pairs out of the first $K$ pairs. By Lemmas~\ref{lemma:i-j-diff-set} and~\ref{lemma:i-j-same-set}, it follows that
\begin{align}
\begin{aligned}
& \left.\mathbb{E}_{U_{1} \sim \mathcal{E}_{e_{1}}}\left[\ldots \mathbb{E}_{U_{K} \sim \mathcal{E}_{e_{K}}}\left[\tr\left[\left(S_{A} \otimes Y^{\otimes 2}\right) U_{1}^{\otimes 2} \ldots U_{K}^{\otimes 2}\left(\left|\mathbf{z}_{1}\right\rangle\left\langle\mathbf{z}_{2}\right| \otimes\left|\mathbf{z}_{3}\right\rangle\left\langle\mathbf{z}_{4}\right|\right) U_{K}^{\dagger \otimes 2} \ldots U_{1}^{\dagger \otimes 2}\right)\right]\right] \ldots\right] \\
& =\left(\prod_{i=1}^{L} F_{e_{i}}\left(\mathbf{z}_{1}, \mathbf{z}_{2}, \mathbf{z}_{3}, \mathbf{z}_{4}\right)\right) \left(\prod_{i=L+1}^{K} H_{e_{i}}\left(\mathbf{z}_{1}, \mathbf{z}_{2}, \mathbf{z}_{3}, \mathbf{z}_{4}\right)\right) \tr\left[\left(S_{A} \otimes Y^{\otimes 2}\right)\left(\left|\mathbf{z}_{1}\right\rangle\left\langle\mathbf{z}_{2}\right| \otimes\left|\mathbf{z}_{3}\right\rangle\left\langle\mathbf{z}_{4}\right|\right)\right] \\
& =\left(\prod_{i=1}^{L} F_{e_{i}}\left(\mathbf{z}_{1}, \mathbf{z}_{2}, \mathbf{z}_{3}, \mathbf{z}_{4}\right)\right) \left(\prod_{i=L+1}^{K} H_{e_{i}}\left(\mathbf{z}_{1}, \mathbf{z}_{2}, \mathbf{z}_{3}, \mathbf{z}_{4}\right)\right) \\
& \times (-1)^{h(z_{1B})+h(z_{3B})} i^{2N_B} \delta_{\mathbf{z}_{3, A}, \mathbf{z}_{2, A}} \delta_{\overline{\mathbf{z}_{1, B}}, \mathbf{z}_{2, B}} \delta_{\mathbf{z}_{1, A}, \mathbf{z}_{4, A}} \delta_{\overline{\mathbf{z}_{3, B}}, \mathbf{z}_{4, B}}
\end{aligned}
\end{align}
Note that because we assume $N_B$ is even, $i^{2N_B}=1$. Thus, we have
\begin{align}
\begin{aligned}
& \left.\mathbb{E}_{U_{1} \sim \mathcal{E}_{e_{1}}}\left[\ldots \mathbb{E}_{U_{L} \sim \mathcal{E}_{e_{L}}}\left[\tr\left[\left(S_{A} \otimes Y^{\otimes 2}\right) U_{1}^{\otimes 2} \ldots U_{L}^{\otimes 2}\left(\left|\mathbf{z}_{1}\right\rangle\left\langle\mathbf{z}_{2}\right| \otimes\left|\mathbf{z}_{3}\right\rangle\left\langle\mathbf{z}_{4}\right|\right) U_{L}^{\dagger \otimes 2} \ldots U_{1}^{\dagger \otimes 2}\right)\right]\right] \ldots\right] \\
& = (-1)^{h(z_{1B})+h(z_{3B})}
\end{aligned}
\end{align}
for all tuples \(\left(\mathbf{z}_{1}, \mathbf{z}_{2}, \mathbf{z}_{3}, \mathbf{z}_{4}\right)\) of bit strings that satisfy
\begin{enumerate}
  \item \(\mathbf{z}_{2, A}=\mathbf{z}_{3, A}, \mathbf{z}_{1, B}= \overline{\mathbf{z}_{2, B}}, \mathbf{z}_{1, A}=\mathbf{z}_{4, A}\), and \(\mathbf{z}_{3, B}= \overline{\mathbf{z}_{4, B}}\)
  \item Either \(\mathbf{z}_{1, A}\oplus \mathbf{z}_{2, A}=\mathbf{0}\) or \(\mathbf{z}_{1, B}\oplus \mathbf{z}_{3, B}=\mathbf{1}\).
  \item Either \(\mathbf{z}_{1, B}\oplus \mathbf{z}_{3, B}=\mathbf{0}\) or \(\mathbf{z}_{1, B}\oplus \mathbf{z}_{3, B}=\mathbf{1}\)
\end{enumerate}
For all other tuples, the expectation value is zero.\\

The latter two conditions can be combined into the following restriction: either \(\mathbf{z}_{1, B}\oplus \mathbf{z}_{3, B}=\mathbf{1}\), or \(\mathbf{z}_{1, B}\oplus \mathbf{z}_{3, B}=\mathbf{0} = \mathbf{z}_{1, A}\oplus \mathbf{z}_{2, A}\). Observe that this is the union of two disjoint cases. \\

\noindent{\textbf{Case 1:}} \(\mathbf{z}_{1, B}\oplus \mathbf{z}_{3, B}=\mathbf{1}\). Then 
\begin{align}
\begin{cases}
    \mathbf{z}_{1,A} = \mathbf{z}_{4,A}\\
    \mathbf{z}_{2,A} = \mathbf{z}_{3,A}\\
    \mathbf{z}_{1,B} = \overline{\mathbf{z}_{2,B}} = \overline{\mathbf{z}_{3,B}} = \mathbf{z}_{4,B}
\end{cases}
\end{align}
and $h(z_{1B})+h(z_{3B}) = N_B$, which is even. There are $d_A^2d_B$ ways to choose the 3 quantities. \\

\noindent{\textbf{Case 2:}} \(\mathbf{z}_{1, B}\oplus \mathbf{z}_{3, B}=\mathbf{0} = \mathbf{z}_{1, A}\oplus \mathbf{z}_{2, A}\). Then 
\begin{align}
\begin{cases}
    \mathbf{z}_{1,A} = \mathbf{z}_{2,A} = \mathbf{z}_{3,A} = \mathbf{z}_{4,A}\\
    \mathbf{z}_{1,B} = \overline{\mathbf{z}_{2,B}} = \mathbf{z}_{3,B} = \overline{\mathbf{z}_{4,B}}
\end{cases}
\end{align}
and $h(z_{1B})+h(z_{3B}) = 2h(z_{1B})$, which is even. There are $d_Ad_B$ ways to choose the 2 quantities. \\

Using~\eqref{eq:trace-bit-string-expansion}, it follows that
\begin{align}
\mathbb{E}_{G \sim \Eunif}\left[\tr\left(G_B\tilde{G}_B\right)\right] = \frac{d_A^2d_B+ d_Ad_B}{(d_Ad_B)^2} = \frac{d_{A}+1}{d_{A} d_{B}}.
\end{align}
\end{proof}

We are finally ready to prove Theorem~\ref{thm:solution-prob}. On a high level, the proof relies on the observation that when the variance of $\tr[G_B^2]$ is sufficiently small, we may separate the expectation value of the quotient $L^\tau(|G\rangle) = \tr[G_B\tilde{G}_B]/\tr[G_B^2]$ into the quotient of the expectation values, giving us an estimate for $p_s(\Eunif)$ via~\eqref{eq:ps-workable-expression}.  
\begin{proof}[Proof of Theorem.~\ref{thm:solution-prob}]
To bound $p_s(\Eunif)$, recall that by~\eqref{eq:ps-workable-expression}, we have
\begin{equation}
    p_s(\Eunif) = \mathbb{E}_{|G\rangle \sim \Eunif}\left[\frac{\tr[G_B\tilde{G}_B]}{\tr[G_B^2]}\right].
\end{equation}
By \citep[Theorem~1]{zhou2022} and \citep[Theorem~3]{zhou2022}, we have that $\E_{G \sim \mathcal{G}}[\tr[G_B^2]] = \frac{d_A+d_B-1}{d_Ad_B}$ and $\var_{G \sim \mathcal{G}}[\tr[G_B^2]] = \frac{(d_A-1)(d_B-1)}{d_Ad_B}$, respectively.  \\

Let $0<r<1$. Define $\mathcal{E}:=\left\{G \in \mathcal{G}:\left|\operatorname{Tr}\left[G_{B}^{2}\right]-\mathbb{E}\left[\operatorname{Tr}\left[G_{B}^{2}\right]\right]\right|<r \mathbb{E}\left[\operatorname{Tr}\left[G_{B}^{2}\right]\right]\right\}$. By Chebyshev's inequality,
\begin{equation*}
\operatorname{Pr}_{G \sim \mathcal{G}}[G \notin \mathcal{E}] \leq \frac{\operatorname{var}\left[\operatorname{Tr}\left[G_{B}^{2}\right]\right]}{r^{2} \mathbb{E}\left[\operatorname{Tr}\left[G_{B}^{2}\right]\right]^{2}}=\frac{\left(d_{A}-1\right)\left(d_{B}-1\right)}{r^{2}\left(d_{A}+d_{B}-1\right)^{2}} \tag{39}
\end{equation*}

For the upper bound, we have
\begin{align*}
\mathbb{E}_{G \sim \mathcal{G}}\left[\frac{\operatorname{Tr}\left[G_{B} \tilde{G}_{B}\right]}{\operatorname{Tr}\left[G_{B}^{2}\right]}\right] & =\frac{1}{|\mathcal{G}|} \sum_{G \in \mathcal{G}} \frac{\operatorname{Tr}\left[G_{B} \tilde{G}_{B}\right]}{\operatorname{Tr}\left[G_{B}^{2}\right]}  \tag{40}\\
& =\frac{1}{|\mathcal{G}|} \sum_{G \in \mathcal{E}} \frac{\operatorname{Tr}\left[G_{B} \tilde{G}_{B}\right]}{\operatorname{Tr}\left[G_{B}^{2}\right]}+\frac{1}{|\mathcal{G}|} \sum_{G \notin \mathcal{E}} \frac{\operatorname{Tr}\left[G_{B} \tilde{G}_{B}\right]}{\operatorname{Tr}\left[G_{B}^{2}\right]}  \tag{41}\\
& \leq \frac{1}{|\mathcal{G}|} \sum_{G \in \mathcal{E}} \frac{\operatorname{Tr}\left[G_{B} \tilde{G}_{B}\right]}{\mathbb{E}\left[\operatorname{Tr}\left[G_{B}^{2}\right]\right](1-r)}+\frac{|\mathcal{G} \backslash \mathcal{E}|}{|\mathcal{G}|}  \tag{42}\\
& \leq \frac{1}{|\mathcal{G}|} \sum_{G \in \mathcal{G}} \frac{\operatorname{Tr}\left[G_{B} \tilde{G}_{B}\right]}{\mathbb{E}\left[\operatorname{Tr}\left[G_{B}^{2}\right]\right](1-r)}+\frac{|\mathcal{G} \backslash \mathcal{E}|}{|\mathcal{G}|}  \tag{43}\\
& =\frac{1}{1-r} \frac{\mathbb{E}\left[\operatorname{Tr}\left[G_{B} \tilde{G}_{B}\right]\right]}{\mathbb{E}\left[\operatorname{Tr}\left[G_{B}^{2}\right]\right]}+\underset{G \sim \mathcal{G}}{\operatorname{Pr}}[G \notin E]  \tag{44}\\
& \leq \frac{1}{1-r} \frac{d_{A}+1}{d_{A}+d_{B}-1}+\frac{\left(d_{A}-1\right)\left(d_{B}-1\right)}{r^{2}\left(d_{A}+d_{B}-1\right)^{2}} \tag{45}\\
\end{align*}

For the lower bound, we have 
\begin{align*}
\mathbb{E}_{G \sim \mathcal{G}}\left[\frac{\operatorname{Tr}\left[G_{B} \tilde{G}_{B}\right]}{\operatorname{Tr}\left[G_{B}^{2}\right]}\right] & =\frac{1}{|\mathcal{G}|} \sum_{G \in \mathcal{G}} \frac{\operatorname{Tr}\left[G_{B} \tilde{G}_{B}\right]}{\operatorname{Tr}\left[G_{B}^{2}\right]}  \tag{46}\\
& =\frac{1}{|\mathcal{G}|} \sum_{G \in \mathcal{E}} \frac{\operatorname{Tr}\left[G_{B} \tilde{G}_{B}\right]}{\operatorname{Tr}\left[G_{B}^{2}\right]}+\frac{1}{|\mathcal{G}|} \sum_{G \notin \mathcal{E}} \frac{\operatorname{Tr}\left[G_{B} \tilde{G}_{B}\right]}{\operatorname{Tr}\left[G_{B}^{2}\right]}  \tag{47}\\
& \geq \frac{1}{|\mathcal{G}|} \sum_{G \in \mathcal{E}} \frac{\operatorname{Tr}\left[G_{B} \tilde{G}_{B}\right]}{\operatorname{Tr}\left[G_{B}^{2}\right]}  \tag{48}\\
& =\frac{1}{|\mathcal{G}|} \sum_{G \in \mathcal{G}} \frac{\operatorname{Tr}\left[G_{B} \tilde{G}_{B}\right]}{\mathbb{E}\left[\operatorname{Tr}\left[G_{B}^{2}\right]\right](1+r)}-\frac{1}{|\mathcal{G}|} \sum_{G \notin \mathcal{E}} \frac{\operatorname{Tr}\left[G_{B} \tilde{G}_{B}\right]}{\mathbb{E}\left[\operatorname{Tr}\left[G_{B}^{2}\right]\right](1+r)}  \tag{49}\\
& \geq \frac{1}{|\mathcal{G}|} \sum_{G \in \mathcal{G}} \frac{\operatorname{Tr}\left[G_{B} \tilde{G}_{B}\right]}{\mathbb{E}\left[\operatorname{Tr}\left[G_{B}^{2}\right]\right](1+r)}-\frac{1}{|\mathcal{G}|} \sum_{G \notin \mathcal{E}} \frac{1}{\mathbb{E}\left[\operatorname{Tr}\left[G_{B}^{2}\right]\right](1+r)}  \tag{50}\\
& =\frac{1}{1+r} \frac{\mathbb{E}\left[\operatorname{Tr}\left[G_{b} \tilde{G}_{B}\right]\right]}{\mathbb{E}\left[\operatorname{Tr}\left[G_{B}^{2}\right]\right]}-\frac{1}{1+r} \frac{\operatorname{Pr} r_{G \sim \mathcal{G}}[G \notin \mathcal{E}]}{\mathbb{E}\left[\operatorname{Tr}\left[G_{B}^{2}\right]\right]}  \tag{51}\\
& \geq \frac{1}{1+r} \frac{d_{A}+1}{d_{A}+d_{B}-1}-\frac{1}{1+r} \frac{d_{A} d_{B}\left(d_{A}-1\right)\left(d_{B}-1\right)}{r^{2}\left(d_{A}+d_{B}-1\right)^{3}} \tag{52}\\
\end{align*}
\end{proof}

\section{Extracting GHZ states from linear cluster states}\label{app:extracting-ghz}
\begin{proof}[Proof of Thm.~\ref{thm:dejong-accordance}]
Our analysis revolves around the results of de Jong \textit{et al}~\cite{deJong2024}, which characterize exactly when it is and is not possible to extract $|\GHZ\rangle_B$ from $|L\rangle_{AB}$ using LC + LPM + CC. We will call any collection of vertices of the form $\{i, i+1, \dots, i+k-1\}$ for some $i\in [N]$ a $k$-island. 
Note that the matrix equation has no solution on the $N$-vertex line graph for measurement configurations with $B$ (the target system) containing a $k$-island for $k\geq 3$, and that GHZ states may be extracted with measurement configurations for which $B$ contains only $1$-islands by~\cite{deJong2024}[Theorem 1]. Thus, 
\begin{align}
    \frac{|S_N\setminus T_N|}{|S_N|} &\leq 
    \frac{\text{de Jong impossible}\wedge \text{matrix eq solvable}}{\text{de Jong impossible}}\\
    &= \frac{\text{de Jong impossible}\wedge \text{matrix eq solvable}\wedge \text{B has 2-island and no larger island}}{\text{de Jong impossible}}\\
    &\leq \frac{\text{matrix equation solvable}\wedge \text{B has 2-island and no larger island}}{\text{de Jong impossible}}\\
    &\leq \frac{\text{B has 2-island and no larger island}}{\text{de Jong impossible}}
\end{align}
We now show that the RHS converges to $0$ as $N\to \infty$. \\

Let $x$ be the number of vertices before the first $2$-island, and $y$ the number of vertices after the last $2$-island. Consider these regions first. We have the following recurrence: for $N$ vertices, the number of ways to choose a subset such that no $2$ vertices are adjacent (i.e., choose $1$-islands) is given by $f(N)=f(N-1)+f(N-2)$, corresponding to not choosing the $N$-th vertex and choosing the $N$-th vertex respectively. Obviously, $f(0)=1, f(1)=2$, so $f(N)=F_{N+2}$, where $F_i$ denotes the $i$-th Fibonacci number. Asymptotically, $f(N) = O(\phi^N)$, where $\phi=\frac{1+\sqrt5}{2}$. \\

In the middle $m:=N-x-y$ vertices, 2 islands are allowed. Let these vertices be indexed by $j$. By definition, $j=1,2,m-1,m$ constitutes the first and last $2$-islands. Let $g(m)$ denote the number of ways to form $2$-islands and choose $1$-islands in between them s.t. $j=1,2,m-1,m$ are $2$-islands. We have the recurrence: 
\begin{equation}
    g(m)= g(m-3)+g(m-4)+ g(m-5)f(1) + g(m-6)f(2) +\cdots
    \end{equation}
\begin{equation}
    = g(m-3) + \displaystyle\sum_{i=4}^{m-2}g(m-i)f(i-4) = \displaystyle\sum_{i=3}^{m-2}g(m-i)F_{i-2}
\end{equation}
We skip vertices $m-1,m$ and additionally $i-2$ more vertices before the next $2$-island. Within the enclosed $i-4$ vertices, we can choose some $1$-islands. In addition, we have the following base cases to the recurrence: $g(2)=1, g(3)=g(4)=0$. \\

We define the generating function for $g(m)$ as 
\begin{equation}
    G(x)=\displaystyle\sum_{m\geq 2}g(m)x^m = g(2)x^2 + g(3)x^3 + g(4)x^4 + \sum_{m\geq 5}\sum_{i=3}^{m-2}g(m-i)F_{i-2}x^m.
\end{equation}
We perform change of variables $k=m-i, l=i-2$. The new bounds are $3\leq i\leq m-2\rightarrow 1\leq l\leq m-4 = k+l-2 \rightarrow l\geq 1,k\geq 2$ and $k+l+2\geq 5 \rightarrow k+l\geq 3$. 
\begin{equation}
    \displaystyle\sum_{m\geq 2}\sum_{i=3}^{m-2}g(m-i)F_{i-2}x^m = \sum_{l\geq 1}\sum_{k\geq2}g(k)F_lx^{k+l+2} = (\sum_{l\geq 1}F_lx^{l+2})(\sum_{k\geq2}g(k)x^k) = (\sum_{l\geq 1}F_lx^{l+2})G(x).
\end{equation}
Recall that the generating function for Fibonacci numbers is $F(x)=\displaystyle\sum_{l\geq 0}F_lx^l= \frac{x}{1-x-x^2}$, so $\displaystyle\sum_{l\geq 1}F_lx^{l+2} = \frac{x^3}{1-x-x^2}$. Let $P(x)=g(2)x^2 + g(3)x^3 + g(4)x^4 = x^2$. Thus, we obtain 
\begin{equation}
    G(x) = P(x) + \frac{x^3}{1-x-x^2}G(x) \implies G(x) = \frac{P(x)(1-x-x^2)}{1-x-x^2-x^3}.
\end{equation}

Let $R\approx 0.544$ be the real solution to $1-x-x^2-x^3=0$. Observe that $G(z)$ analytic on open disk of radius $R$ centered at origin. By Taylor expansion, $G(z) = \displaystyle\sum_{n\geq0} \frac{G^{(n)}(0)}{n!}z^n \forall z:|z|<R$. Thus, $g(N)=\frac{G^{(N)}(0)}{N!}$. By Cauchy inequality, $|G^{(n)}(0)|\leq \frac{n!M_r}{r^n} \forall 0<r<R, M_r:= \displaystyle\max_{|z|=r}|G(z)|$. Pick $r$ s.t. $1/r < 2$; this is possible because $1/R < 2$. Denote $a:=1/r > 1/R \approx 1.839$. Then $g(n)\leq M_ra^n = O(a^n)$. \\

Observe that the number of bipartitions where the matrix equation is solvable and $B$ has a $2$-island and no larger island is 
\begin{equation}
    \displaystyle\sum_{x,y\geq 1; n-x-y\geq 2} f(x-1)f(y-1)g(N-x-y) = \sum_{x,y\geq 1; N-x-y\geq 2} O(\phi^x)O(\phi^y)O(a^{N-x-y}).
\end{equation}
Although it appears that for each multiplicand in the sum, the big-O notation only gives an asymptotic bound, we can establish a common bound for all $x$ (or $y$ or $N-x-y$). This can be done because for large enough $x$, we have a bound from $O(\phi^x)$ (the big-O notation), and for $x$ not "large enough" (there is finitely many such $x$), we have a bound for each $x$ explicitly via $\frac{F_x}{\phi^x}$ (similarly, $\frac{F_y}{\phi^y}$ and $\frac{g(N-x-y)}{a^{N-x-y}}$). Then we simply take the max of these bounds. \\

Let $k=x+y$. The sum becomes 
\begin{equation}
    \displaystyle\sum_{k=2}^{N-2} (k-1) O(\phi^k)O(a^{N-k}) \leq \sum_{k=2}^{N-2} NO(\phi^k)O(a^{N-k}) = O(Na^N\sum_{k=2}^{N-2}(\frac{\phi}{a})^k).
\end{equation}
Because $\frac{\phi}{a} < 1$, $\displaystyle\sum_{k=2}^{N-2}(\frac{\phi}{a})^k < \sum_{k=0}^\infty(\frac{\phi}{a})^k = \frac{1}{1-\frac{\phi}{a}} = O(1)$. Thus, 
\begin{equation}
    \displaystyle O\left(Na^N\sum_{k=2}^{N-2}(\frac{\phi}{a})^k\right) = O(Na^N).
\end{equation}

Finally, to find the number of measurement configurations where GHZ extraction is impossible by de Jong, we consider the total number of measurement configurations $2^n-2$ minus those where GHZ extraction is possible via LC + LPM +CC. Per~\cite{deJong2024}, the only cases where GHZ extraction is possible are 1-islands in $B$ only, measurement configurations where 2-islands in $B$ are either arranged as $2-\dots-2$, $\dots - 2$, or $2-\dots$, where the dots denote the presence of 1-islands only, and $B$ consisting of a 3-island only. From our prior analysis, the number of such cases is bounded asymptotically by $\Theta(F_n)$. Thus, the number of bipartitions where GHZ extraction is impossible is $\Theta(2^n-F_n)=\Theta(2^n)$. Thus, 
\begin{equation}
    \frac{\text{B has 2-island and no larger island}}{\text{de Jong impossible}} = \frac{O(na^n)}{\Theta(2^n)} = O\left(n\left(\frac{a}{2}\right)^n\right) \rightarrow 0,
\end{equation}
as claimed.
\end{proof}

\section{Measure Concentration Phenomena for LE and EA}\label{app:concentration-haar}
\begin{center}
\begin{tikzpicture}[x=4.2cm, y=-2.2cm, scale=0.6]
    % nodes
    \node (Vairogs) [item] at (2,0) {Results from~\cite{Vairogs2024}};
    \node (Prop6)   [item] at (1,1) {Prop.~\ref{prop:norm-additivity}};
    \node (Lemma9)  [item] at (3,1) {Lem.~\ref{lem:Fv-Lipschitz-wrt-v}};
    \node (Thm8)    [item] at (1,2) {Thm.~\ref{thm:basis-net-thm}};
    \node (Lemma10) [item] at (3,2) {Lem.~\ref{lem:tau-bar-lipschitz}};
    \node (Lemma15) [item] at (2,3) {Lem.~\ref{lem:ltl-neighbor-bound}, Lem.~\ref{lem:Fv-lipschitz-wrt-Psi}};
    \node (Thm16)   [item] at (1,4) {Thm.~\ref{thm:local-rB}};
    \node (Cor17)   [item] at (1,5) {Cor.~\ref{cor:local-concentration}};
    \node (Lemma19) [item] at (3, 4) {Lem.~\ref{lem:neighbor-bound-F-v}};
    \node (Thm20) [item] at (3, 5) {Thm.~\ref{thm:global-rB}};
    \node (Cor21) [item] at (3, 6) {Cor.~\ref{cor:global-concentration}};
    % arrows
    \draw [arrow] (Thm20) -- (Cor21);
    \draw [arrow] (Lemma9)  -- (Lemma10);
    \draw [arrow] (Lemma15) -- (Thm16);   
    \draw [arrow] (Thm16)   -- (Cor17);   
    \draw[arrow] (Lemma10) -- (Lemma15);  
    \draw [arrow] (Vairogs) -- (Lemma15);
    \draw [arrow] (Thm8) -- (Thm16);
    \draw [arrow] (Prop6) -- (Thm8);
    \draw[arrow] (Vairogs) -- (Lemma15);
    \draw[arrow] (Thm8) -- (Lemma15);
    \draw[arrow] (Lemma15) -- (Lemma19);
    \draw[arrow] (Lemma10) -- (Lemma19);
    \draw[arrow] (Lemma19) -- (Thm20);
    % boxes
    \node[highlight, fit=(Thm16)(Cor17)(Lemma19)(Thm20)(Cor21), label=below:\textbf{Main results}]{};
    \node[highlight, fit=(Thm16)(Cor17), label=below:(Local)]{};
    \node[highlight, fit=(Lemma19)(Thm20)(Cor21), label=below:(Global)]{};
\end{tikzpicture}
\end{center}

\subsection{Basis $\varepsilon$-nets}

\begin{proposition} \label{prop:trace-to-inner}
    For all states $|u\rangle, |v\rangle \in \mathcal S(\mathcal H)$, the following inequality holds:
    \begin{equation}
        \||u\rangle\langle u| - |v\rangle\langle v|\|_1 = 2\sqrt{1-|\langle u|v\rangle|^2}.
    \end{equation}
\end{proposition}

\begin{proof}
    Let $u := |u\rangle \langle u|$ and $A = u-v$ for simplicity. Obviously, $\operatorname{rank}(A) \le 2$, which implies that there are at most $2$ non-zero eigenvalues of $A$: let us denote them by $\lambda_1$ and $\lambda_2$. Then
    \begin{equation}
        \lambda_1 + \lambda_2 = \tr[A] = \tr[u-v] = \tr[u] - \tr[v] = 1-1 = 0\implies \lambda_2 = -\lambda_1.
    \end{equation}
    Observe that
    \begin{align}
        2\lambda_1^2 = \lambda_1^2 + \lambda_2^2 = \tr[A^2] &= \tr[(u-v)^2] \\
        &= \tr[u^2 - uv - vu + v^2] \\
        &= \tr[u^2] - \tr[uv] - \tr[vu] + \tr[v^2] \\
        &= 2 - 2|\langle u|v\rangle|^2,
    \end{align}
    which gives $|\lambda_1| = \sqrt{1-|\langle u|v\rangle|^2}$. Obviously, $A$ is Hermitian, so $\lambda_1,\lambda_2\in \mathbb R$. Finally,
    \begin{equation} \textstyle
        \|u-v\|_1 = \| A\|_1 = \tr[\sqrt{A^\dagger A}] = \tr[\sqrt{A^2}] = \sqrt{\lambda_1^2} + \sqrt{\lambda_2^2} = |\lambda_1| + |\lambda_2| = 2|\lambda_1| \nonumber = 2\sqrt{1-|\langle u|v\rangle|^2}.
    \end{equation}
\end{proof}

\begin{corollary} \label{cor:trace-hilbert-link}
    For all states $|u\rangle, |v\rangle \in \mathcal S(\mathcal H)$, the following inequality holds:
    \begin{equation}
        \||u\rangle\langle u| - |v\rangle\langle v|\|_1 \le 2\| |u\rangle - |v\rangle\|_2.
    \end{equation}
\end{corollary}

\begin{proof} By Prop.~\ref{prop:trace-to-inner},
    \begin{align}
        \| |u\rangle\langle u| - |v\rangle\langle v|\|_1 = 2\sqrt{1-|\langle u|v\rangle|^2} \label{eq:trace-to-inner} &\le 2\sqrt{2 - 2|\langle u|v\rangle|} \\
        &\le 2\sqrt{2-2\operatorname{Re}(\langle u|v\rangle)} \\
        &= 2\sqrt{\langle u|u\rangle - \langle u|v\rangle -\langle v|u\rangle + \langle v|v\rangle} \\
        &= 2\sqrt{(\langle u|-\langle v|)(|u\rangle - |v\rangle)} \\
        &= 2\| |u\rangle -|v\rangle\|_2.
    \end{align}
\end{proof}

\begin{proposition}\label{prop:norm-additivity}
    Let $\mathcal{H}_1, \mathcal{H}_2$ be Hilbert spaces. For all orthonormal bases $\beta_1, \beta_2 \in \mathcal{C}(\mathcal{H}_1)$ and $\gamma_1, \gamma_2 \in \mathcal{C}(\mathcal{H}_2)$, the following inequality holds:
    \begin{equation}
        \|\beta_1 \otimes \gamma_1 - \beta_2 \otimes \gamma_2\|_B\leq \|\beta_1 - \beta_2\|_B+ \|\gamma_1 - \gamma_2\|_B.
    \end{equation}
\end{proposition}
\begin{proof}
    Let $n$ and $m$ be the dimension of $\mathcal H_1$ and $\mathcal H_2$, respectively. We may write $\beta_1 = \{|\varphi_i\rangle\}_{i=1}^n$, $\beta_2 = \{|\varphi_i'\rangle\}_{i=1}^n$, $\gamma_1 = \{|\eta_j\rangle\}_{j=1}^m$, and $\gamma_2 = \{|\eta_j'\rangle\}_{j=1}^m$. Let $i\in[n]$ and $j\in[m]$ be integers such that $\|\beta_1 \otimes \gamma_1 - \beta_2\otimes \gamma_2\|_B= \|\varphi_i \otimes \eta_j - \varphi_i' \otimes \eta_j'\|_1.$ Then
    \begin{align}
        \|\beta_1 \otimes \gamma_1 - \beta_2\otimes \gamma_2\|_B &= \|\varphi_i \otimes \eta_j - \varphi_i' \otimes \eta_j'\|_1 \\
        &= \|\varphi_i \otimes \eta_j - \varphi_i' \otimes \eta_j + \varphi_i' \otimes \eta_j - \varphi_i' \otimes \eta_j' \|_1 \\
        &\leq \|\varphi_i \otimes \eta_j - \varphi_i' \otimes \eta_j\|_1 + \|\varphi_i' \otimes \eta_j - \varphi_i' \otimes \eta_j' \|_1 \\ 
        &= \|\varphi_i - \varphi_i'\|_1 \cdot \|\eta_j\|_1 + \|\varphi_i'\|_1 \cdot \|\eta_j - \eta_j'\|_1 \\
        &= \|\varphi_i - \varphi_i'\|_1 + \|\eta_j - \eta_j'\|_1 \\
        &\leq \|\beta_1 - \beta_2\|_B+ \|\gamma_1 - \gamma_2\|_B\,, \quad \text{by Def.~\ref{def:basis-norm}.}
    \end{align}
\end{proof}

\begin{proposition} \label{prop:basis-net-1q}
    For $\varepsilon\in (0,1)$, there exists a basis $[(1+2\sqrt2)\sqrt{\varepsilon}]$-net $\mathcal{N}$ on $\mathcal{C}(\mathbb{C}^2)$ with $|\mathcal{N}|\leq (5/\varepsilon)^8$.
\end{proposition}

\begin{proof}
    Lemma II.4 of \cite{Hayden_2004} states that there exists an $\varepsilon$-net $\mathcal{N}_s$ on $\mathbb{C}^2$ with $|\mathcal{N}_s|\leq (5/\varepsilon)^4$. Let 
    \begin{equation} \label{eq:construction-of-N}
        \mathcal{N} := \left\{ \left\{|\eta_1\rangle, \frac{(I - \eta_1)|\eta_2\rangle}{\sqrt{\langle \eta_2|(I - \eta_1)|\eta_2\rangle}} \right\}: |\eta_1\rangle, |\eta_2\rangle \in \mathcal{N}_s \text{ and } \langle\eta_2|(I - \eta_1)|\eta_2\rangle>0\right\}.
    \end{equation}

    Notice that $\frac{(I - \eta_1)|\eta_2\rangle}{\sqrt{\langle \eta_2|(I - \eta_1)|\eta_2\rangle}}$ is the normalized projection of $|\eta_2\rangle$ onto the orthogonal complement of $|\eta_1\rangle\langle\eta_1|$. Hence, we verified that every element in $\mathcal N$ is an orthonormal basis, i.e., $\mathcal N \subset \mathcal C(\C^2)$.
    
    Let $\beta= \{|\varphi_1\rangle, |\varphi_2\rangle \} \in \mathcal{C}(\mathbb{C}^2)$ be arbitrary. By the definition of an $\varepsilon$-net \cite{Hayden_2004}, there exist $|\eta_1\rangle , |\eta_2\rangle \in \mathcal{N}_s$ such that $\|\eta_1 -\varphi_1\|_1\le \varepsilon$ and $\|\eta_2 - \varphi_2\|_1 \leq \varepsilon$. Let $|\eta_2'\rangle := \frac{(I - \eta_1)|\eta_2\rangle}{\sqrt{\langle \eta_2|(I - \eta_1)|\eta_2\rangle}}$ and $\gamma := \{|\eta_1\rangle, |\eta_2'\rangle\} \in \mathcal N$. Def.~\ref{def:basis-norm} gives
    \begin{equation}
        \|\beta - \gamma\|_B = \max\{\|\varphi_1 - \eta_1\|_1,\, \|\varphi_2 - \eta_2'\|_1\}.
    \end{equation}
    We already have $\|\varphi_1 - \eta_1\|_1\le \varepsilon$. Observe that
    \begin{align}
        \|\eta_1 - \eta_2\|_1 &= \| (\eta_1- \varphi_1) - (\eta_2 - \varphi_2) + (\varphi_1-\varphi_2) \|_1 \label{eq:eta-dist-1}\\
        &\ge \| \varphi_1 - \varphi_2\|_1 - \|(\eta_1- \varphi_1) - (\eta_2 - \varphi_2)\|_1 \\
        &\ge \|\varphi_1 - \varphi_2\|_1 - (\|(\eta_1- \varphi_1)\|_1 + \|(\eta_2 - \varphi_2)\|_1) \\
        &\ge \textstyle 2\sqrt{1-|\langle \varphi_1 | \varphi_2\rangle|^2} - 2\varepsilon \\
        &= 2-2\varepsilon. \label{eq:eta-dist-2}
    \end{align}
    The following relations give an upper bound for $\| \varphi_2 - \eta_2'\|_1$:
    \begin{align}
        \|\varphi_2 - \eta_2' \|_1 &=  \| \varphi_2 - \eta_2 + \eta_2 - \eta_2' \|_1 \\
        &\leq \| \varphi_2 - \eta_2 \| + \| \eta_2 - \eta_2' \|_1 \\
        &\le \textstyle \varepsilon + 2 \sqrt{1 - | \langle \eta_2 | \eta_2'\rangle |^2} & \text{by Prop.~\ref{prop:trace-to-inner}} \\
        &= \varepsilon + 2\sqrt{1 - \langle \eta_2 |(I-\eta_1)| \eta_2\rangle} \\
        &= \varepsilon + 2|\langle \eta_1 | \eta_2 \rangle| \\
        &= \varepsilon + 2 \sqrt{1 - (\| \eta_1 - \eta_2 \|_1/2)^2} & \text{by Prop.~\ref{prop:trace-to-inner}} \\
        &\leq \varepsilon + 2 \sqrt{1 - (1 - \varepsilon)^2} & \text{by Eq.~\eqref{eq:eta-dist-1}$-$\eqref{eq:eta-dist-2}}\\
        &\leq (1 + 2 \sqrt{2})\sqrt\varepsilon.
    \end{align}
    Thus, $\|\beta - \gamma\|_B$ is bounded above by $\max\{\varepsilon, (1+2\sqrt2)\sqrt\varepsilon\} = (1+2\sqrt2)\sqrt\varepsilon$. By Def.~\ref{def:basis-epsnet}, $\mathcal N$ is indeed a basis $[(1+2\sqrt2)\sqrt\varepsilon]$-net over $\mathcal C(\C^2)$.

    As for the cardinality bound, Eq.~\eqref{eq:construction-of-N} implies $\mathcal N \subseteq N_s \times N_s$. Thus, $|\mathcal N| \le |\mathcal N_s|^2 = (5/\varepsilon)^8$.
\end{proof}

\begin{proof}[Proof of Thm.~\ref{thm:basis-net-thm}]
    Let $\delta := (\varepsilon/n)^2/(2\sqrt2+1)^2$ so that $(2\sqrt2+1)\sqrt\delta = \varepsilon/n$. Lemma II.4 of \cite{Hayden_2004} gives that there exists a basis $(\varepsilon/n)$-net $\mathcal{M}$ on $\mathcal{C}(\mathbb{C}^2)$ whose cardinality is bounded above by $(5/\delta)^8 = [5(2\sqrt2 + 1)^2n^2/\varepsilon^2]^8$. Set $\mathcal{N} := \{\gamma_1 \otimes \dots \otimes \gamma_n: \gamma_i \in \mathcal{M}\}.$ By definition, $\mathcal N \subset \mathcal P((\mathbb C^2)^{\otimes n})$. Let $\beta_1 \otimes \dots \otimes \beta_n$
    be an arbitrary basis in $\mathcal P((\mathbb C^2)^{\otimes n})$ and $\gamma_1 \otimes \dots \otimes \gamma_n$ be an arbitrary member of $\mathcal N$. Prop.~\ref{prop:norm-additivity} gives the following inequality
    \begin{equation}
        \left\lVert \bigotimes_{i=1}^n \nolimits \gamma_i - \bigotimes_{i=1}^n \nolimits \beta_i\right\rVert_B \,\le\, \sum_{i=1}^n \nolimits \lVert \gamma_i - \beta_i\rVert_B. \label{eq:8777888787}
    \end{equation}
    Since each $\gamma_i$ comes from the ($\varepsilon/n$)-net $\mathcal M$, by definition we can choose $\gamma_i$ for each $i\in[n]$ such that $\lVert \gamma_i - \beta_i\rVert_B \le \varepsilon/n$ and therefore --- continuing from Eq.~\eqref{eq:8777888787} --- $\sum_{i=1}^n \nolimits \lVert \gamma_i - \beta_i\rVert_B \le n \cdot (\varepsilon/n) = \varepsilon$, which shows that $\mathcal N$ is indeed a basis $\varepsilon$-net. Finally, our choice of $\mathcal N$ implies that $|\mathcal N| = |\mathcal M|^n \le [5(2\sqrt2 + 1)^2n^2/\varepsilon^2]^{8n}.$
\end{proof}

\subsection{Entanglement of Assistance}
\begin{lemma}\label{lem:Fv-lipschitz-wrt-Psi}
    For arbitrary $|v\rangle\in \mathcal S(\mathcal H_A)$ and $|\Psi\rangle, |\Psi'\rangle \in \mathcal S(\mathcal H_A\otimes \mathcal H_B)$, the following probability bound holds:
    \begin{align}
        |F_v(\Psi) - F_v(\Psi')| \leq (4\sqrt{2}+2)\||\Psi\rangle - |\Psi'\rangle\|_2.
    \end{align}
\end{lemma}

\begin{proof}
    Let $p_u$, $p_u'$, $\varphi_u$, $\varphi_u'$ denote $p_u(\Psi)$, $p_u(\Psi')$, $|M_u(\Psi)\rangle$, $|M_u(\Psi')\rangle$ for $|u\rangle \in \mathcal S(\mathcal H_A)$, respectively. Then
    \begin{align}
        |F_v(\Psi) - F_v(\Psi')| &= |p_v\tau(\varphi_v) - p_v'\tau(\varphi_v')| \\
        &= |p_v\tau(\varphi_v) - p_v\tau(\varphi_v') + p_v\tau(\varphi_v') - p_v'\tau(\varphi_v')| \\
        &= |p_v (\tau(\varphi_v) - \tau(\varphi_v')) + (p_v - p_v')\tau(\varphi_v')| \\
        &\le |p_v (\tau(\varphi_v) - \tau(\varphi_v'))| + |(p_v - p_v')\tau(\varphi_v')|
    \end{align}
    Let $|i\rangle \in \mathcal C(\mathcal H_A)$ denote a basis that contains $|v\rangle$. By \citep[Lem.~16]{Vairogs2024}, we have
    \begin{gather}
        \textstyle |p_v(\tau(\varphi_v) - \tau(\varphi_v'))| \le \sum_i p_i|\tau(\varphi_i) - \tau(\varphi_i')| \le 2\sqrt2\|\Psi - \Psi'\|_1 \le 4\sqrt2\||\Psi\rangle - |\Psi'\rangle\|_2 \\
        \textstyle |(p_v - p_v')\tau(\varphi_v')| \le \sum_i |p_i - p_i'| \le \|\Psi - \Psi'\|_1 \le 2\||\Psi\rangle - |\Psi'\rangle\|_2.
    \end{gather}
\end{proof}

\begin{lemma} \label{lem:Fv-Lipschitz-wrt-v}
    For all $|v\rangle, |w\rangle \in \mathcal S(\mathcal{H}_A)$, the following probability bound holds:
    \begin{equation}
        \left| F_v(\Psi) - F_w(\Psi) \right| \leq \sqrt{2} d_B \left\||v\rangle\langle v| - |w\rangle\langle w|\right\|_1.
    \end{equation}
\end{lemma}

\begin{proof}
    For all $|u\rangle \in \mathcal S(\mathcal H_A)$ and for all $\theta \in \R$, we have
    \begin{equation}
        F_u(e^{i\theta}|\Psi\rangle) = \big|\langle P_u(e^{i\theta}|\Psi\rangle)|\wt P_u(e^{i\theta}|\Psi\rangle)\rangle\big| = \big|e^{-2i\theta} \langle P_u(\Psi)|\wt P_u(\Psi)\rangle\big| = F_u(\Psi).
    \end{equation}
    Thus, we may WLOG assume $\langle v|w\rangle =|\langle v|w\rangle|$. Observe the following relations (for simplicity we omit $\Psi$):
    \begin{align}
        \big| F_v - F_w\big|  &= \big||\langle P_v|\wt{P}_v\rangle| - |\langle P_w|\wt{P}_w\rangle|\big| &\text{by Eq.~\eqref{eq:fvpvmv}} \\
        &\le \big| \langle P_v| \wt{P}_v \rangle - \langle P_w|\wt{P}_w \rangle \big| \\
        &= \big| \langle P_v| \wt{P}_v \rangle - \langle P_w  | \wt{P}_v \rangle + \langle P_w  | \wt{P}_v \rangle -\langle P_w|\wt{P}_w \rangle \big| \\
        &= \big| (\langle P_v | - \langle P_w |)| \wt{P}_v  \rangle + \langle P_w | (|\wt{P}_v \rangle - |\wt{P}_w  \rangle) \big| \\
        &\leq \big| (\langle P_v | - \langle P_w |)| \wt{P}_v  \rangle\big| + \big|\langle P_w | (|\wt{P}_v \rangle - |\wt{P}_w  \rangle) \big| \\
        &\le \| |P_v\rangle - |P_w\rangle \|_2 \| |\wt P_v\rangle \|_2 + \| |P_w\rangle \|_2 \| |\wt P_v\rangle - |\wt P_w\rangle \|_2 \label{eq:i-am-here-cuz-cauchy} &\text{by Cauchy-Schwartz} \\
        &= \| |P_v\rangle - |P_w\rangle \|_2 \| |P_v\rangle \|_2 + \| |P_w\rangle \|_2 \| | P_v \rangle - |P_w  \rangle \|_2 \\
        &\le 2\||P_v\rangle - |P_w\rangle\|_2. \label{eq:wait-for-bound}
    \end{align}
    For simplicity let $|\alpha\rangle = |v\rangle - |w\rangle$ and $|i\rangle \in \mathcal C(\mathcal H_B)$. Let us continue from Eq.~\eqref{eq:wait-for-bound}:
    \begin{align}
        \big|F_v - F_w\big| \le 2\||P_v\rangle - |P_w\rangle\| &= 2\|(\langle\alpha| \otimes I_B)|\Psi\rangle\|_2 \\
        &= \textstyle 2\sqrt{\langle\Psi| (|\alpha\rangle\langle \alpha|_A \otimes I_B) |\Psi\rangle} \\
        &= \textstyle 2\sqrt{\sum_{i=1}^{d_B} \langle\Psi|(|\alpha\rangle\langle \alpha|_A \otimes |i\rangle\langle i|_B)|\Psi\rangle} \\
        &\le \textstyle 2\sum_{i=1}^{d_B} \sqrt{\langle\Psi|(|\alpha\rangle\langle \alpha|_A \otimes |i\rangle\langle i|_B)|\Psi\rangle} \\
        &= \textstyle 2\sum_{i=1}^{d_B} \big|(\langle \alpha|_A\langle i|_B)|\Psi\rangle\big| \\
        &\le \textstyle 2\sum_{i=1}^{d_B} \||\alpha\rangle_A|i\rangle_B\|_2 \cdot \||\Psi\rangle\|_2 &\text{by Cauchy-Schwartz} \\
        &= \textstyle 2\sum_{i=1}^{d_B} \||\alpha\rangle\|_2 \cdot \||i\rangle\|_2 \cdot \||\Psi\rangle\|_2 \\
        &= 2d_B\||\alpha\rangle\|_2 \\
        &= 2d_B\sqrt{(\langle v| - \langle w|)(|v\rangle - |w\rangle)} \\
        &= 2d_B\sqrt{2 - 2\Re(\langle v|w\rangle)} \\
        &= 2d_B\sqrt{2 - 2|\langle v|w\rangle|} &\because \langle v|w\rangle = |\langle v|w\rangle| \\
        &\le 2d_B\sqrt{2 - 2|\langle v|w\rangle|^2} &\because |\langle v|w\rangle|\le 1 \\
        &= \sqrt 2d_B\||v\rangle\langle v| - |w\rangle\langle w|\|_1 &\text{by Prop.~\ref{prop:trace-to-inner}.}
    \end{align}
\end{proof}

\begin{lemma} \label{lem:neighbor-bound-F-v}
    Let $\varepsilon, \delta>0$ be such that $\varepsilon - \sqrt{2}d_B\delta>0$. Given a fixed $|v\rangle \in \mathcal S(\mathcal{H}_A)$, we have the following probability bound
    \begin{equation}
        \Pr_{|\Psi\rangle \sim \mu_H}\left(F_{\varphi_{\max}}(|\Psi\rangle) \geq \tfrac{K}{d_A} + \varepsilon ~\text{ and }~ \||\varphi_{\max}\rangle\langle \varphi_{\max}| - |v\rangle\langle v|\|_1 \leq \delta \right) \leq 2\exp\left(- \tfrac{2d_Ad_B(\varepsilon - \sqrt{2}d_B\delta)^2}{9\pi^3 (4\sqrt{2}+2)^2}\right),
    \end{equation}
    which implies that as long as $|\varphi_{\max}\rangle$ and $|v\rangle$ are close enough, the value of $F_{\varphi_{\max}}(\Psi)$ will be close to $K/d_A$. (Recall that $K := \sqrt{2/(d_B + 1)}$ as defined in Eq.~\eqref{eq:K}.)
\end{lemma}

\begin{proof}
    Let us claim that 
    \begin{align} F_{\varphi_{\max}}(\Psi) \ge \tfrac{K}{d_A} + \varepsilon ~\text{ and }~ \||\varphi_{\max}\rangle\langle\varphi_{\max}| - |v\rangle\langle v|\|_1 \le \delta ~\implies~ F_v(\Psi) \ge \tfrac{K}{d_A} + \varepsilon - \sqrt2 d_B\delta. \end{align}
    Proof:
    \begin{align}
        F_v(\Psi) &= F_{\varphi_{\max}}(\Psi) - (F_{\varphi_{\max}}(\Psi) - F_v(\Psi)) \label{eq:fvp6a0-1} \\
        &\ge F_{\varphi_{\max}}(\Psi) - |F_{\varphi_{\max}}(\Psi) - F_v(\Psi)| \\
        &\ge F_{\varphi_{\max}}(\Psi) - \sqrt2 d_B\||\varphi_{\max}\rangle\langle\varphi_{\max}| - |v\rangle\langle v|\|_1 &\text{by Lem.~\ref{lem:Fv-Lipschitz-wrt-v}} \\
        &\ge \tfrac{K}{d_A} + \varepsilon - \sqrt2 d_B \delta. \label{eq:fvp6a0-2}
    \end{align}
    Moreover, if $\beta = \{|\varphi_i\rangle\}_{i=1}^{d_A}$ is an arbitrary basis of $\mathcal H_A$, by \citep[Lem.~22]{Vairogs2024} we have
    \begin{align}
        \E_{|\Psi'\rangle \sim \mu_H} [F_v(\Psi')] &= \frac{1}{d_A}\sum_{i=1}^{d_A} \nolimits \E_{|\Psi'\rangle \sim \mu_H} [F_{\varphi_i}(\Psi')] \label{eq:haar-invariance-1944-1} \\
        &= \frac{1}{d_A}\E_{|\Psi'\rangle \sim \mu_H} [\overline \tau_\beta(\Psi')] \\
        &\le \frac{K}{d_A}. \label{eq:haar-invariance-1944-2}
    \end{align}
    where Eq.~\eqref{eq:haar-invariance-1944-1} is obtained because of the left/right invariance of the Haar measure: Suppose $|w\rangle = U|v\rangle$ for an arbitrary unitary matrix $U\in \mathrm U(d_A)$. Then
    \begin{align}
        F_w(|\Psi'\rangle) &= \big|\langle P_w(\Psi')| \widetilde P_w(\Psi')\rangle\big| \\
        &= \big|\langle \Psi'|(|w\rangle_A \otimes I_B) \sigma_y^{\otimes N_B}(\langle w^*|_A \otimes I_B)|\Psi'^*\rangle\big| \\
        &= \big|\langle \Psi'|(|w\rangle\langle w^*| \otimes \sigma_y^{\otimes N_B})|\Psi'^*\rangle\big| \\
        &= \big| \langle\Psi'|(U|v\rangle\langle v^*|U^T \otimes \sigma_y^{\otimes N_B})|\Psi'^*\rangle \big| \\
        &= \big| \langle\Psi'|(U\otimes I)(|v\rangle \langle v^*| \otimes \sigma_y^{\otimes N_B})(U^T \otimes I)|\Psi'^*\rangle \big| \\
        &= F_v \big( (U^\dagger \otimes I)|\Psi'\rangle \big).
    \end{align}
    Hence,
    \begin{align}
        \Pr_{|\Psi\rangle \sim \mu_H} &\left(F_{\varphi_{\max}}(|\Psi\rangle) \geq \tfrac{K}{d_A} + \varepsilon ~\text{ and }~ \| |\varphi_{\max}\rangle\langle\varphi_{\max}| - |v\rangle\langle v| \|_1 \le \delta \right) \\
        &\le \Pr_{|\Psi\rangle \sim \mu_H} \left(F_v(\Psi) \ge \tfrac{K}{d_A} + \varepsilon - \sqrt2 d_B \delta\right) & \text{by Eqs.~\eqref{eq:fvp6a0-1}$-$\eqref{eq:fvp6a0-2}} \\
        &\le \Pr_{|\Psi\rangle \sim \mu_H} \left(F_v(\Psi) \ge \E_{|\Psi'\rangle \sim \mu_H}[F_v(\Psi')] + \varepsilon - \sqrt2 d_B \delta\right) & \text{by Eqs.~\eqref{eq:haar-invariance-1944-1}$-$\eqref{eq:haar-invariance-1944-2}} \\
        &= \Pr_{|\Psi\rangle \sim \mu_H} \left(F_v(\Psi) - \E_{|\Psi'\rangle \sim \mu_H}[F_v(\Psi')] \ge \varepsilon - \sqrt2 d_B \delta\right) \\
        &\le 2\exp\left(- \tfrac{2d_Ad_B(\varepsilon - \sqrt{2}d_B\delta)^2}{9\pi^3 (4\sqrt{2}+2)^2}\right). & \text{by Levy's lemma \cite{Mele2024}}
    \end{align}
\end{proof}

\begin{proof}[Proof of Thm.~\ref{thm:global-rB}]
    For all $|\Psi\rangle \in \mathcal H_A \otimes \mathcal H_B$, the following inequality holds:
    \begin{equation} \label{eq:190j4-1}
        L^\tau_{\text{global}}(\Psi) = \max_{\{|\varphi_i\rangle\} \,\in\, \mathcal C(\mathcal H_A)} \sum_{i=1}^{d_A} \nolimits F_{\varphi_i}(\Psi) \le d_AF_{\varphi_{\max}}(\Psi).
    \end{equation}
    Let $\|\varphi_{\max} - v\|_1 := \||\varphi_{\max}\rangle \langle \varphi_{\max}| - |v\rangle \langle v|\|_1$ in the following. 
    Also, by \citep[Lem.~II.4]{Hayden_2004}, we can choose an $(\varepsilon/2\sqrt2d_Ad_B)$-net $\mathcal N$ on $\mathcal H_A$ with $|\mathcal N| \le (10\sqrt2 d_Ad_B/\varepsilon)^{2d_A}$, which suggests
    \begin{equation} \label{eq:smth-must-happen0}
        \Pr_{|\Psi\rangle \sim \mu_H} \left( \bigvee_{|v\rangle \in \mathcal H_A} \nolimits \|\varphi_{\max} - v\|_1 \le \tfrac{\varepsilon}{2\sqrt2 d_Ad_B} \right) = 1.
    \end{equation}
    Then
    \begin{align}
        &\text{\hspace*{-2em}} \Pr_{|\Psi\rangle \sim \mu_H} \big( L^\tau_{\text{global}}(\Psi) \ge K + \varepsilon \big) \le \Pr_{|\Psi\rangle \sim \mu_H} \left( F_{\varphi_{\max}}(\Psi) \ge \tfrac{K}{d_A} + \tfrac{\varepsilon}{d_A}\right) \qquad \text{by Eq.~\eqref{eq:190j4-1}} \\
        &= \Pr_{|\Psi\rangle \sim \mu_H} \left( \bigcup_{|v\rangle \in \mathcal N} \nolimits \left\{F_{\varphi_{\max}}(\Psi) \ge \tfrac{K}{d_A} + \tfrac{\varepsilon}{d_A} ~\wedge~ \|\varphi_{\max} - v\|_1 \le \tfrac{\varepsilon}{2\sqrt2 d_Ad_B} \right\}\right) \qquad \text{by Eq.~\eqref{eq:smth-must-happen0}} \\
        &\le \sum_{|v\rangle \in \mathcal N} \Pr_{|\Psi\rangle \sim \mu_H} \left( F_{\varphi_{\max}}(\Psi) \ge \tfrac{K}{d_A} + \tfrac{\varepsilon}{d_A} ~\wedge~ \|\varphi_{\max} - v\|_1 \le \tfrac{\varepsilon}{2\sqrt2 d_Ad_B} \right) \\
        &\le |\mathcal N| \cdot 2\exp\left(-\tfrac{d_B\varepsilon^2}{18\pi^3(4\sqrt2+2)^2d_A}\right) \qquad \text{by Lem.~\ref{lem:neighbor-bound-F-v}} \\
        &\le 2\left(\tfrac{10\sqrt2 d_Ad_B}{\varepsilon}\right)^{2d_A} \exp\left(-\tfrac{d_B\varepsilon^2}{18\pi^3(4\sqrt2+2)^2d_A}\right).
    \end{align}
\end{proof}

\subsection{Localizable Multipartite Entanglement}

\begin{lemma} \label{lem:tau-bar-lipschitz}
    Let $\beta = \{|\varphi_i\rangle\}_{i=1}^{d_A}, \gamma = \{|\eta_i\rangle\}_{i=1}^{d_A} \in \mathcal C(\mathcal H_A)$ be ordered orthonormal bases. Then 
    \begin{equation}
        |\overline{\tau}_\beta(|\Psi\rangle) - \overline{\tau}_{\gamma}(|\Psi\rangle)| \leq \sqrt{2}d_Ad_B\|\beta - \gamma\|_B.
    \end{equation}
\end{lemma}
\begin{proof}
    \begin{align}
        |\overline{\tau}_\beta(|\Psi\rangle) - \overline{\tau}_{\gamma}(|\Psi\rangle)| &= \bigg| \sum_{i = 1}^{d_A} \nolimits F_{\varphi_i}(\Psi) - \sum_{i = 1}^{d_A} \nolimits F_{\eta_i} (\Psi) \bigg| \\
        &= \bigg| \sum_{i = 1}^{d_A} \nolimits F_{\varphi_i} (\Psi) - F_{\eta_i} (\Psi) \bigg| \\
        &\leq \sum_{i = 1}^{d_A} \nolimits | F_{\varphi_i} (\Psi) - F_{\eta_i}(\Psi) | \label{eq:abs-val-inside}\\
        & \leq \sum_{i = 1}^{d_A} \nolimits \sqrt{2} d_B \| | \varphi_i \rangle \langle \varphi_i | - | \eta_i \rangle \langle \eta_i | \|_1 \qquad \text{by Lem.~\ref{lem:Fv-Lipschitz-wrt-v}} \\
        &= \sqrt{2} d_B \sum_{i = 1}^{d_A} \nolimits \| | \varphi_i \rangle \langle \varphi_i | - |\eta_i \rangle \langle \eta_i | \|_1 \\
        &\leq \sqrt{2} d_A d_B \| \beta - \gamma \|_B. \label{eq:B-norm-result}
    \end{align}
\end{proof}

\begin{lemma} \label{lem:ltl-neighbor-bound}
    For any $|\Psi\rangle \in \mathcal S(\mathcal H_A \otimes \mathcal H_B)$, choose $\beta_{\max} \in \mathcal P(\mathcal H_A)$ to be the ordered basis such that $\overline \tau_\beta(\Psi) \le \overline \tau_{\beta_{\max}}(\Psi)$ for all $\beta \in \mathcal P(\mathcal H_A)$. By Eq.~\eqref{eq:ltl-def-alt}, $\ltl(\Psi) = \overline \tau_{\beta_{\max}}(\Psi)$.
    Let $\varepsilon, \delta>0$ be such that $\varepsilon - \sqrt{2}d_Ad_B\delta>0$ and $\gamma \in \mathcal{P}(\mathcal{H}_A)$ be some fixed basis. Then the following probability bound holds:
    \begin{equation} \label{eq:lemma-11-main}
        \Pr_{|\Psi\rangle \sim \mu_H} \Big(\overline\tau_{\beta_{\max}}(\Psi) \geq K + \varepsilon ~\text{ and }~ \|\beta_{\max} - \gamma\|_B \leq \delta \Big) \leq 2\exp\left(- \tfrac{2d_Ad_B(\varepsilon - \sqrt{2}d_Ad_B\delta)^2}{9\pi^3 (4\sqrt{2}+2)^2}\right)
    \end{equation}
\end{lemma}

\begin{proof} Let us claim that
    \begin{align}
        \overline\tau_{\beta_{\max}}(\Psi) \geq K + \varepsilon ~\text{ and }~ \|\beta_{\max} - \gamma\|_B \leq \delta ~\implies~ \overline \tau_\gamma(\Psi) \ge K + \varepsilon - \sqrt{2}d_Ad_B\delta.
    \end{align}
    Proof:
    \begin{align}
        \overline{\tau}_{\gamma} (\Psi) &= \overline{\tau}_{\beta_{\max}}(\Psi) - (\overline{\tau}_{\beta_{\max}}(\Psi) - \overline{\tau}_{\gamma}(\Psi)) \label{eq:tau-gamma-bond-1} \\
        &\geq \overline{\tau}_{\beta_{\max}}(\Psi)- | \overline{\tau}_{\beta_{\max}}(\Psi) - \overline{\tau}_{\gamma}(\Psi) | \\
        &\ge \overline{\tau}_{\beta_{\max}}(\Psi) - \sqrt{2}d_Ad_B \| \beta_{\max} - \gamma \|_B  &\text{by Lem.~\ref{lem:tau-bar-lipschitz}} \\
        &\ge K + \varepsilon - \sqrt{2}d_Ad_B\delta. \label{eq:tau-gamma-bond-2}
    \end{align}
    Hence,
    \begin{align}
        \Pr_{|\Psi\rangle \sim \mu_H} &\Big(\overline\tau_{\beta_{\max}}(\Psi) \geq K + \varepsilon ~\text{ and }~ \|\beta_{\max} - \gamma\|_B \leq \delta \Big) \\
        &\le \Pr_{|\Psi\rangle \sim \mu_H} \Big(\overline{\tau}_{\gamma}(\Psi)\geq K + \varepsilon - \sqrt{2}d_Ad_B\delta \Big) &\text{by Eq.~\eqref{eq:tau-gamma-bond-1}$-$\eqref{eq:tau-gamma-bond-2}} \\
        &\leq \Pr_{|\Psi\rangle \sim \mu_H} \left( \overline{\tau}_{\gamma}(\Psi) \geq \E_{|\Psi'\rangle \sim \mu_H}[\overline{\tau}_{\gamma}(\Psi')] + \varepsilon - \sqrt{2}d_Ad_B\delta \right) &\text{by \citep[Lem.~22]{Vairogs2024}} \\
        &= \Pr_{|\Psi\rangle \sim \mu_H} \left( \overline{\tau}_{\gamma}(\Psi) - \E_{|\Psi'\rangle \sim \mu_H}[\overline{\tau}_{\gamma}(\Psi')] \ge \varepsilon - \sqrt{2}d_Ad_B\delta \right) \\
        &\leq \Pr_{|\Psi\rangle \sim \mu_H} \left( \left|\overline{\tau}_{\gamma}(\Psi) - \E_{|\Psi'\rangle \sim \mu_H}[\overline{\tau}_{\gamma}(\Psi')]\right| \ge \varepsilon - \sqrt{2}d_Ad_B\delta \right). \label{eq:levys-lemma}
    \end{align}

    By \citep[Lem.~16]{Vairogs2024}, suppose that a concave, increasing function $f:[0,\infty) \to [0,\infty)$ satisfies $|\tau(|\psi\rangle) - \tau(|\psi'\rangle)| \le  f(\|\psi - \psi'\|_1)$ and $\tau(|\psi\rangle) \le f(\|\psi\|_1)$ for all normalized $|\psi\rangle, |\psi'\rangle$. Then for all $|\Psi\rangle, |\Psi'\rangle \in \mathcal S(\mathcal H_A \otimes \mathcal H_B)$, the following inequality holds:
    \begin{equation}
        |\overline{\tau}_{\beta} (|\Psi\rangle) -\overline{\tau}_{\beta}(|\Psi'\rangle)| \leq f(2 \| \Psi - \Psi' \|_1) + \| \Psi - \Psi' \|_1.
    \end{equation}
    Also, we can easily show that $f(x) = \sqrt2 x$ satisfies the aforementioned conditions. Thus, we have the following
    \begin{align}
        | \overline{\tau}_{\gamma}(|\Psi\rangle) - \overline{\tau}_{\gamma}(|\Psi' \rangle ) | \leq (1+2\sqrt2) \|\Psi - \Psi'\|_1 \leq (2 + 4\sqrt2) \|\Psi - \Psi'\|_2,
    \end{align}
    The above implies that $\overline \tau_\gamma(\Psi)$ is Lipschitz continuous with Lipschitz constant $4\sqrt2 + 2$, and hence from Eq.~\eqref{eq:levys-lemma}, we may use Levy's lemma \cite{Mele2024} to finish the proof.
\end{proof}

\begin{proof}[Proof of Thm.~\ref{thm:local-rB}]
    By Thm.~\ref{thm:basis-net-thm}, there exists a basis $(\varepsilon/2\sqrt2d_Ad_B)$-net $\mathcal N$ on $\mathcal P(\mathcal H_A)$ with
    \begin{equation} \label{eq:ncard-bond}
        |\mathcal N| \le \left(\tfrac{5(1+2\sqrt2)^2N_A^2}{(\varepsilon/2\sqrt2d_Ad_B)^2}\right)^{8N_A} = \left(\tfrac{40(1+2\sqrt2)^2N_A^2d_A^2d_B^2}{\varepsilon^2}\right)^{8N_A}.
    \end{equation}
    Then by Def.~\ref{def:basis-epsnet} there must be $\gamma \in \mathcal N$ such that $\|\beta_{\max} - \gamma\| \le \varepsilon/2\sqrt2d_Ad_B$, which suggests
    \begin{equation} \label{eq:smth-must-happen}
        \Pr_{|\Psi\rangle \sim \mu_H} \left(\bigvee_{\gamma \in \mathcal N} \nolimits \|\beta_{\max} - \gamma\| \le \tfrac{\varepsilon}{2\sqrt2 d_Ad_B}\right) = 1.
    \end{equation}
    Hence,
    \begin{align}
        &\text{\hspace*{-4em}}\Pr_{|\Psi\rangle \sim \mu_H} \big(\ltl(\Psi) \geq K + \varepsilon \big) \nonumber \\
        = ~~&\Pr_{|\Psi\rangle \sim \mu_H} \big(\overline\tau_{\beta_{\max}}(\Psi) \geq K + \varepsilon \big) \\
        = ~~&\Pr_{|\Psi\rangle \sim \mu_H} \left(\bigcup_{\gamma \in \mathcal N} \nolimits \left\{\overline\tau_{\beta_{\max}}(\Psi) \ge K + \varepsilon ~\wedge~ \|\beta_{\max} - \gamma\| \le \tfrac{\varepsilon}{2\sqrt2 d_Ad_B}\right\}\right)\,, \qquad \text{by Eq.~\eqref{eq:smth-must-happen}} \\
        \le ~~&\sum_{\gamma \in \mathcal N} \nolimits \Pr_{|\Psi\rangle \sim \mu_H}\left(\overline\tau_{\beta_{\max}}(\Psi) \ge K + \varepsilon ~\wedge~ \|\beta_{\max} - \gamma\| \le \tfrac{\varepsilon}{2\sqrt2 d_Ad_B}\right) \\
        \le ~~&|\mathcal N| \cdot 2\exp\left(-\tfrac{d_Ad_B\varepsilon^2}{18\pi^3(2+4\sqrt2)^2}\right)\,, \qquad \text{by Lem.~\ref{lem:ltl-neighbor-bound}} \\
        \le ~~&2\left(\tfrac{40(1+2\sqrt2)^2N_A^2d_A^2d_B^2}{\varepsilon^2}\right)^{8N_A}\exp\left(-\tfrac{d_Ad_B\varepsilon^2}{18\pi^3(2+4\sqrt2)^2}\right)\,, \qquad \text{by Eq.~\eqref{eq:ncard-bond}.}
    \end{align}   
\end{proof}

\section{Simulation Details}
\label{app:sim-details}

\subsection{Nonisomorphic Graphs and Bipartitions Thereof}
\label{app:noniso_sim}
Due to the large number of possible graphs and bipartitions, it is only feasible to enumerate small $n$, specifically $n \leq 10$, where $n$ is the number of vertices in the graph. We use the \href{https://users.cecs.anu.edu.au/~bdm/data/graphs.html}{online database} of nonisomorphic graphs due to Brendan McKay. Immediately we observe that number of nonisomorphic graphs for each reported by this database are consistent with OEIS1349 (number of graphs with n unlabeled vertices). \\

For $n \leq 8$, the number of nonisomorphic graphs and possible bipartitions are small enough to perform a complete enumeration of all bipartitions. For $9\leq n\leq 10$, we continue to sample randomly from all nonisomorphic graphs and all bipartitions. \\

An issue is that bipartitions may be isomorphic, which again causes uneven probability distribution. We define bipartition isomorphism as follows. Consider two graphs of $n$ vertices $G_1(V_1,E_1), G_2(V_2,E_2)$. Then 2 bipartitions $(G_1,X_1), X_1 \subset V_1, (G_2,X_2), X_2 \subset V_2$ are isomorphic if $|x_1| = |x_2|$ and there exists relabeling of vertices $p:\{1,\cdots,n\}\rightarrow \{1,\cdots,n\}$ s.t. $(G_1,x_1)$ and $(p(G_2),p(x_2))$ have identical adjacency matrices for the $X$ subgraph, $G-X$ subgraph, and for edges between $X$ and $G-X$. \\

To generate nonisomorphic bipartitions of a graph for $n \leq 8$, we can enumerate all relabelings $p$ and see if the resulting adjacency matrices has already been seen. Note that two nonisomorphic graphs cannot produce isomorphic bipartitions; otherwise, the 3 adjacency matrices uniquely determine the adjacency matrix of the original graph, so two graphs would be isomorphic by the relabeling from bipartition isomorphism. As a result, we only need to compare each relabeling's adjacency matrix tuple against relabeling tuples of the same (nonisomorphic) graph. \\

For $9\leq n\leq 10$, we use \citep[Thm.~1]{bandit} to give an approximation argument. Due to this theorem, we can use the following procedure to obtain a close approximation of the true proportion of solvable bipartitions for a fixed bipartition size. For each nonisomorphic graph, we take exactly 1 (random) bipartition. Note this is a random sample of all nonisomorphic bipartitions. For our case, a specific element of the sample has value 1 if it is solvable and 0 otherwise, thus $b=1$. The sample size are $261080, 11716571$ for $n=9,10$ respectively, which is sufficient for an error of $<0.004$ for $n=9$ and $<0.0005$ for $n=10$ with probability of $>99.98\%$. This accuracy is sufficient to see a trend in the data: the percentage solvable appears to be a flipped logistic (sigmoid) function over $k$, with a midpoint around $k=n/2$. 

\subsection{Graph generation for graph families}\label{app:graph-families}

Generation of the families of graphs was done by taking all possible bipartitions for graphs of a certain family, by fixing the number of vertices. So, if we fix the number of vertices to be $10$ and the family to be a cycle, for example, then all possible bipartitions were considered for the $10-$cycle, and the probability of solution was calculated as the ratio of these bipartitions that give a solution to the total number of bipartitions of a $10-$cycle. The different families of graphs with these probabilities for graphs with $16$ vertices are shown in Fig. \ref{fig:alternative-ensembles}b.

\end{document}